\documentclass[runningheads,envcountsect,envcountsame]{llncs} %

\usepackage{pslatex}
\usepackage[T1]{fontenc}

\usepackage{amsmath}
\usepackage{amsfonts}
\usepackage{amssymb}
\usepackage{bm}
\usepackage{gastex}
\usepackage{theorem}
\usepackage{stmaryrd} 
\usepackage[basic]{complexity}
\usepackage{url}

\usepackage{ifthen}

\newboolean{short_version}
\setboolean{short_version}{false}

\ifthenelse{\boolean{short_version}}{
\usepackage{tikz}
}{
\usepackage{tikz}
}

\newcommand{\lexico}{\text{lex}}
\newcommand{\leqlex}{\leq_{\lexico}}
\newcommand{\ltlex}{<_{\lexico}}

\newcommand{\obracew}[2]{{\overset{#2}{\overbrace{#1}}}}

\newcommand{\Id}{{\mathit{Id}}}
\newcommand{\Ccal}{\mathcal{C}}

\newcommand{\Nat}{{\mathbb{N}}} 
\renewcommand{\emptyset}{\varnothing}
\renewcommand{\epsilon}{\varepsilon} 
\renewcommand{\setminus}{\smallsetminus} 
\newcommand{\sat}{\models}      

\newcommand{\egdef}{\stackrel{\mbox{\begin{scriptsize}def\end{scriptsize}}}{=}}
\newcommand{\equivdef}{\stackrel{\mbox{\begin{scriptsize}def\end{scriptsize}}}{\Leftrightarrow}}

\newcommand{\tuple}[1]{\langle #1 \rangle}
\newcommand{\size}[1]{{\mathopen{\mid}#1\mathclose{\mid}}}

\newcommand{\subword}{\sqsubseteq}
\newcommand{\invsubword}{\sqsupseteq}

\theoremstyle{plain}
\newtheorem{fact}[theorem]{Fact}

\newdimen\arrowsize
\pgfarrowsdeclare{arcs}{arcs}
{
  \arrowsize=0.4pt
  \advance\arrowsize by .5\pgflinewidth
  \pgfarrowsleftextend{-4\arrowsize-.5\pgflinewidth}
  \pgfarrowsrightextend{.5\pgflinewidth}
}
{
  \arrowsize=0.4pt
  \advance\arrowsize by .5\pgflinewidth
  \pgfsetdash{}{0pt} 
  \pgfsetroundjoin   
  \pgfsetroundcap    
  \pgfpathmoveto{\pgfpoint{-4\arrowsize}{4\arrowsize}}
  \pgfpatharc{190}{260}{5\arrowsize}
  \pgfpatharc{100}{170}{5\arrowsize}
  \pgfusepathqstroke
}
\def\ins{\mathrel{\tikz[baseline]\path[-arcs,line width=0.4pt](0pt,.6ex) edge (1em,.6ex);}}
\def\del{\mathrel{\tikz[baseline]\path[-arcs,dash pattern=on 1.5pt off 1.5pt,line width=0.4pt](0pt,.6ex) edge (1em,.6ex);}}
\newcommand{\INS}{{\mathop{\mathbf{ins}}}}
\newcommand{\DEL}{{\mathop{\mathbf{del}}}}
\newcommand{\Dom}{{\mathop{\mathit{Dom}}}}
\newcommand{\stepgram}{\Rightarrow}

\newcommand{\renaming}{\mathsf{R}}
\newcommand{\B}{\boxempty}

\begin{document}


\title{Toward a Compositional Theory of Leftist Grammars and Transformations\thanks{Work supported by the Agence Nationale
  de la Recherche, grant ANR-06-SETIN-001.}}
\author{P.\ Chambart \and Ph.\ Schnoebelen}
\institute{LSV, ENS Cachan, CNRS\\
	61, av.\ Pdt.\ Wilson, F-94230 Cachan, France
}
\maketitle



\begin{abstract}
Leftist grammars [Motwani \textit{et al.}, STOC 2000] are special semi-Thue
systems where symbols can only insert or erase to their left. 
We develop a theory of leftist grammars \emph{seen as word transformers}
as a tool toward rigorous analyses of their computational
power. 
Our main contributions in this first paper are (1) constructions proving
that leftist transformations are closed under compositions and transitive
closures, and (2) a proof that bounded reachability is $\NP$-complete even
for leftist grammars with acyclic rules.
\end{abstract}



\section{Introduction}
\label{sec-intro}

Leftist grammars were introduced by Motwani \textit{et al.} to study
accessibility and safety in protection systems~\cite{motwani2000}. In this
framework, leftist grammars are used to show that restricted accessibility
grammars have decidable accessibility problems (unlike the more general
access-matrix model).

Leftist grammars are both surprisingly simple and surprisingly complex.
Simplicity comes from the fact that they only allow rules of the form
``$a\rightarrow ba$'' and ``$cd\rightarrow d$'' where a symbol inserts,
resp.\	 erases, another symbol to its left \emph{while remaining unchanged}.
But the combination of insertion and deletion rules makes leftist grammars
go beyond context-sensitive grammars, and the decidability result comes
with a high complexity-theoretical price~\cite{jurdzinski2008}. Most of
all, what is surprising is that apparently leftist grammars had not been
identified as a relevant computational formalism until 2000.

The known facts on leftist grammars and their computational and expressive
power are rather scarce. Motwani \textit{et al.} show that it is decidable
whether a given word can be derived (accessibility) and whether all
derivable words belong to a given regular language
(safety)~\cite{motwani2000}. Jurdzi\'nski and Lory\'s showed that leftist
grammars can define languages that are not
context-free~\cite{jurdzinski2007c} while leftist grammars restricted to
acyclic rules are less expressive since they can only recognize regular
languages. Then Jurdzi\'nski showed a $\PSPACE$ lower bound for
accessibility in leftist grammars~\cite{jurdzinski2007b}, before improving
this to a nonprimitive-recursive lower bound~\cite{jurdzinski2008}.

Jurdzi\'nski's results rely on encoding classical computational structures
(linear-boun\-ded automata~\cite{jurdzinski2007b} and Ackermann's
function~\cite{jurdzinski2008}) in leftist grammars. Devising such
encodings is difficult because leftist grammars are very hard to control.
Thus, for computing Ackermann's function, devising the encoding is actually
not the hardest part: the harder task is to prove that the constructed
leftist grammar cannot behave in unexpected ways. In this regard, the
published proofs are necessarily incomplete, hard to follow, and hard to
fully acknowledge. The final results and intermediary lemmas cannot easily be
adapted or reused.

\paragraph{Our Contribution.}
We develop a compositional theory of leftist grammars and leftist
transformations (i.e., operations on strings that are computed by leftist
grammars) that provides fundamental tools for the analysis of their
computational power. Our main contributions are effective constructions for
the composition and the transitive closure of leftist transformations. The
correctness proofs for these constructions are based on new definitions
(e.g., for greedy derivations) and associated lemmas.

A first application of the compositional theory is given
in Section~\ref{sec-np-hard} where we prove the $\NP$-completeness of bounded
reachability questions, even when restricted to acyclic leftist grammars.

A second application, and the main reason for this paper, is our
forthcoming construction proving that leftist grammars can simulate lossy
channel systems and ``compute'' all multiply-recursive transformations and
nothing more (based on~\cite{CS-lics08}), thus providing a precise measure
of their computational power. Finally, after our introduction of Post's
Embedding Problem~\cite{CS-fsttcs07,CS-omegapep}, leftist grammars are
another basic computational model that will have been shown to capture
exactly the notion of multiply-recursive computation.

As further comparison with earlier work, we observe that, of course, the
complex constructions in~\cite{jurdzinski2007b,jurdzinski2008} are built
modularly. However, the modularity is not made fully explicit in these
works, the interfacing assumptions are incompletely stated, or are mixed with
the details of the constructions, and correctness proofs cannot be given in
full.

\paragraph{Outline of the Paper.} Basic notations and definitions are
recalled in Section~\ref{sec-basics}. Section~\ref{sec-leftist} defines
leftist grammars and proves a generalized version of the completeness of
greedy derivations. Sections~\ref{sec-ltr} introduces leftist transformers
and their sequential compositions.
Section~\ref{sec-simple-transformer} specializes on the ``simple''
transformers that we use in Section~\ref{sec-np-hard} for our encoding of
$\3SAT$. Finally Section~\ref{sec-transitive} shows that so-called
``anchored'' transformers are closed under the transitive closure
operation, this in an effective way. 
\ifthenelse{\boolean{short_version}}{
For lack of space, several proofs have been omitted in this extended abstract: they can be found in
the long version of this  paper, freely available at the \url{arXiv}.
}{
}




\section{Basic Definitions and Notations}
\label{sec-basics}

\paragraph{Words.}
We use $x,y,u,v,w,\alpha,\beta,\ldots$ to denote words, i.e., finite
strings of symbols taken from some alphabet. Concatenation is denoted
multiplicatively with $\epsilon$ (the empty word) as neutral element, and
the length of $x$ is denoted $\size{x}$.

The congruence on words generated by the equivalences $a\approx aa$ (for
all symbols $a$ in the alphabet) is called the \emph{stuttering equivalence} and
is also denoted $\approx$: every word $x$ has a minimal and canonical
stuttering-equivalent $x'$ obtained by repeatedly eliminating symbols in
$x$ that are adjacent to a copy of themselves.

We say that $x$ is a \emph{subword} of $y$, denoted $x\subword y$, if $x$
can be obtained by deleting some symbols (an arbitrary number, at arbitrary
positions) from $y$. We further write $x\subword_{\Sigma}y$ when all the
symbols deleted from $y$ belong to $\Sigma$ (NB: we do not require
$y\in\Sigma^*$), and let $\invsubword$ denote the inverse relation $\subword^{-1}$.

\paragraph{Relations and Relation Algebra.}
We see a relation $\mathrel{R}$ between two sets $X$ and $Y$ as a set of
pairs, i.e., some $R\subseteq X\times Y$. We write $x\mathrel{R}y$
rather than $(x,y)\in R$. Two relations $R$ and $R'$ can be composed,
denoted multiplicatively with $R.R'$, and defined by
$
     x\mathrel{(R.R')}y \equivdef \exists z. \bigl(x\mathrel{R}z\;\wedge\; z\mathrel{R'}y\bigr).
$

The union $R+R'$, also denoted $R\cup R'$, is just the set-theoretic union.
$R^n$ is the $n$-th power $R.R\ldots R$ of $R$ and $R^{-1}$ is the inverse
of $R$: $x\mathrel{R^{-1}}y\equivdef y \mathrel{R}x$. The transitive closure
$\bigcup_{n=1,2,\ldots}R^n$ of $R$ assumes $Y=X$ and is denoted $R^+$,
while its reflexive-transitive closure is $R^+\cup \Id_X$, denoted $R^*$.

Below we often use notations from relation algebra to state simple
equivalences. E.g., we write ``$R=R'$'' and ``$R\subseteq S$'' rather than
``$x\mathrel{R}y$ iff $x\mathrel{R'}y$'' and ``$x\mathrel{R}y$ implies $x\mathrel{S}y$''. Our proofs often
rely on well-known basic laws from relation algebra, like
$(R.R')^{-1}=R'^{-1}.R^{-1}$, or $(R+R').R''=R.R''+ R'.R''$, without
explicitly stating them.



\section{Leftist Grammars}
\label{sec-leftist}

A \emph{leftist grammar} (an LGr) is a triple $G=(\Sigma,P,g)$ where
$\Sigma\cup\{g\}=\{a,b,\ldots\}$ is a finite \emph{alphabet},
$g\not\in\Sigma$ is a \emph{final symbol} (also called ``\emph{axiom}''),
and $P=\{r,\ldots\}$ is a set of production rules that may be
\emph{insertion rules} of the form $a\rightarrow ba$, and \emph{deletion
  rules} of the form $cd\rightarrow d$. For simplicity, we forbid rules
that insert or delete the axiom $g$ (this is no loss of
generality~\cite[Prop.\  3]{jurdzinski2007c}).

Leftist grammars are not context-free (deletions are contextual), or even
context-sensitive (deletions are not length-preserving). For our purposes,
we consider them as string rewrite systems, more precisely semi-Thue
systems. Writing $\Sigma_g$ for $\Sigma\cup\{g\}$, the rules of $P$ define
a 1-step rewrite relation in the standard way: for $u,u'\in\Sigma_g^*$, we
write $u\stepgram^{r,p} u'$ whenever $r$ is some rule
$\alpha\rightarrow\beta$, $u$ is some $u_1\alpha u_2$ with
$\size{u_1\alpha}=p$ and $u'=u_1\beta u_2$. We often write shortly
$u\stepgram^{r} u'$, or even $u\stepgram u'$, when the position or the rule
involved in the step can be left implicit. On the other hand, we sometimes
use a subscript, e.g., writing $u\stepgram_G v$, when the
underlying grammar has to be made explicit.

A \emph{derivation} is a sequence $\pi$ of consecutive rewrite steps, i.e.,
is some $u_0\stepgram^{r_1,p_1} u_1 \stepgram^{r_2,p_2} u_2
\cdots\stepgram^{r_n,p_n} u_n$, often abbreviated as $u_0\stepgram^n u_n$,
or even $u_0\stepgram^* u_n$.
A subsequence $(u_{i-1}\stepgram^{r_i,p_i}u_i)_{i=m,m+1,\ldots,l}$ of $\pi$ is a
\emph{subderivation}.
As with all semi-Thue systems, steps (and derivations) are closed
under adjunction: if $u\stepgram u'$ then $vuw\stepgram vu'w$.

Two derivations $\pi_1=(u\stepgram^* u')$ and $\pi_2=(v\stepgram^* v')$ can be
concatenated in the obvious way (denoted $\pi_1.\pi_2$) if $u'=v$. They are
\emph{equivalent}, denoted $\pi_1\equiv \pi_2$, if they have same
extremities, i.e., if $u=v$ and $u'=v'$.

We say that $u\in\Sigma^*$ is \emph{accepted by $G$} if there is a
derivation of the form $ug\stepgram^*g$ and we write $L(G)$ for the set of
accepted words, i.e., the language recognized by $G$.

We say that $I\subseteq \Sigma^*$ is an \emph{invariant} for an LGr
$G=(\Sigma,P,g)$ if $u\in I$ and $ug\stepgram vg$ entail $v\in I$. Knowing
that $I$ is an invariant for $G$ is used in two symmetric ways: (1) from
$u\in I$ and $ug\stepgram^* vg$ one deduces $v\in I$, and (2) from
$ug\stepgram^* vg$ and $v\not\in I$ one deduces $u\not\in I$.

\subsection{Graphs and Types for Leftist Grammars}
\label{ssec-types}

When dealing with LGr's, it is convenient to write insertion rules under
the simpler form ``$a\ins b$'', and deletion rules as ``$d\del c$'',
emphasizing the fact that $a$ (resp.\ $d$) is not modified during the
insertion of $b$ (resp.\  the deletion of $c$) on its left. For $a\in
\Sigma_g$, we let $\INS(a)\egdef\{b~|~P\ni (a\ins b)\}$ and
$\DEL(a)\egdef\{b~|~P\ni (a\del b)\}$ denote the set of symbols that can be
inserted (respectively, deleted) by $a$. We write $\INS^+(a)$ for the
smallest set that contains $b$ and $\INS^+(b)$ for all $b\in\INS(a)$, while
$\DEL^+(b)$ is defined similarly. We say that $a$ is \emph{inactive} in a
LGr if $\DEL(a)\cup\INS(a)=\emptyset$.
\\

It is often convenient to view LGr's in a
graph-theoretical way.
Formally, the \emph{graph} of $G=(\Sigma,P,g)$ is the directed graph
$\tau_G$ having the symbols from $\Sigma_g$ as vertices and the rules from
$P$ as edges (coming in two kinds, insertions and deletions).
Furthermore, we often
decorate such graphs with extra bookkeeping annotations.

We say that $G$ ``\emph{has type $\tau$}'' when $\tau_G$ is a sub-graph of
$\tau$. Thus a ``type'' is just a restriction on what are the allowed
symbols and rules between them. Types are often given schematically,
grouping symbols that play a similar role into a single vertex. 
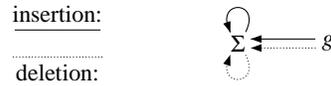
\begin{figure}[bthp]
\centering
{\setlength{\unitlength}{0.60mm}
\begin{picture}(80,10)(0,0)

\gasset{Nframe=n,Nw=6,Nh=6,Nmr=999,loopdiam=5}

\node(s)(55,5){$\Sigma$}
\node(sg)(75,5){$g$}

\drawedge[syo=1,eyo=1](sg,s){}
\drawedge[syo=-1,eyo=-1,dash={0.2 0.5}0](sg,s){}
\drawloop[loopangle=90,loopCW=n](s){}
\drawloop[loopangle=-90,dash={0.2 0.5}0](s){}

\node(aux1)(5,5){}
\node(aux2)(25,5){}
\gasset{AHnb=0}
\drawedge[syo=3,eyo=3,ELside=r](aux2,aux1){\text{insertion:}}
\drawedge[syo=-3,eyo=-3,dash={0.2 0.5}0,ELside=l](aux2,aux1){\text{deletion:}}

\end{picture}}
\caption{Universal type (schematically).}
\label{fig-type-universal}
\end{figure}
For
example, Fig.~\ref{fig-type-universal} displays schematically the type
(parametrized by the alphabet) observed by all LGr's.

\subsection{Leftmost, Pure and Eager Derivations}
\label{ssec-greedy}

We speak informally of a ``\underline{letter}'', say $a$, when we really mean
``an occurrence of the \underline{symbol} $a$'' (in some word).
Furthermore, we follow letters along steps $u\stepgram v$, identifying the
letters in $u$ and the corresponding letters in $v$. Hence a ``letter'' is
also a sequence of occurrences in consecutive words along a derivation.

A letter $a$ is a $n$-th \emph{descendant} of another letter $b$ (in the
context of a derivation) if $a$ has been inserted by $b$ (when $n=1$), or
by a $(n-1)$-th descendant of $b$.

Given a step $u\stepgram^{r,p}v$, we say that the $p$-th letter in $u$,
written $u[p]$, is the \emph{active letter}: the one that inserts, or
deletes, a letter to its left. This is often emphasized by writing the step
under the form $(u=)u_1au_2\stepgram u'_1au_2(=v)$ (assuming $u[p]=a$).

A letter is \emph{inert} in a derivation if it is not active \emph{in
any step} of the derivation. A set of letters is inert if it only contains
inert letters. A derivation is \emph{leftmost} if every step
$u_1au_2\stepgram u'_1au_2$ in the derivation is such that $u_1$ is inert
in the rest of the derivation. 

A letter is \emph{useful} in a derivation $\pi=(u\stepgram^* v)$ if it
belongs to $u$ or $v$, or if it inserts or deletes a useful letter along
$\pi$. This recursive definition is well-founded: since letters only insert
or delete to their left, the ``inserts-or-deletes'' relation between
letters is acyclic. A derivation $\pi$ is \emph{pure} if all letters in
$\pi$ are useful. Observe that if $\pi$ is not pure, it necessarily inserts
at some step some letter $a$ (called a \emph{useless letter}) that stays
inert and will eventually be deleted.

A derivation is \emph{eager} if, informally, deletions occur as soon as
possible. Formally, $\pi=(u_0\stepgram^{r_1,p_1} u_1 \cdots
\stepgram^{r_n,p_n} u_n)$ is not eager if there is some	$u_{i-1}$ of  the form
$w_1baw_2$ where $b$ is inert in the rest of $\pi$ and is eventually
deleted,  where $P$ contains the rule $a \del b$, and where  $r_i$ is not a
deletion rule.\footnote{Eagerness does not require that $r_i$ deletes $b$: 
other deletions are allowed, only insertions are forbidden.}

A derivation is \emph{greedy} if it is leftmost, pure and eager. Our
definition generalizes~\cite[Def.~4]{jurdzinski2007b}, most notably because
it also applies to derivations $ug\stepgram^* vg$ with nonempty $v$. Hence
a subderivation $\pi'$ of $\pi$ is leftmost, eager, pure, or greedy,
when $\pi$ is. 

The following proposition generalizes~\cite[Lemma~7]{jurdzinski2007b}.
\begin{proposition}[Greedy derivations are sufficient]
\label{prop-greedy}
Every derivation $\pi$ has an equivalent greedy derivation $\pi'$.
\end{proposition}
\begin{proof}
With a derivation $\pi$ of the form $u_0\stepgram^{r_1,p_1} u_1
\stepgram^{r_2,p_2} u_2 \cdots \stepgram^{r_n,p_n}u_n$, we
associate its \emph{measure} $\mu(\pi)\egdef\tuple{n,p_1,\ldots,p_n}$, 
a $(n+1)$-tuple of numbers.
Measures are linearly ordered with the lexicographic ordering,
giving rise to a quasi-ordering, denoted $\leq_\mu$, between derivations.
A derivation is called \emph{$\mu$-minimal} if any
equivalent derivation has greater or equal measure.

We can now prove Prop.~\ref{prop-greedy} 
\ifthenelse{\boolean{short_version}}{
along the following lines:
first prove that
every derivation has a $\mu$-minimal equivalent,
then show that $\mu$-minimal derivations are
greedy.
}{
along the following lines
(see Appendix~\ref{app-prop-greedy} for full details): 
first prove that
every derivation has a $\mu$-minimal equivalent
(Lemma~\ref{lem-has-equiv-mini}), then show that $\mu$-minimal derivations are
greedy (Lemma~\ref{lem-mini-is-greedy}).
}
\qed
\end{proof}

Observe that $\leq_\mu$ is compatible with concatenation of derivations: if
$\pi_1\leq_\mu\pi_2$ then $\pi.\pi_1.\pi'\leq_\mu \pi.\pi_2.\pi'$ when
these concatenations are defined. Thus
any subderivation of a $\mu$-minimal derivation is
$\mu$-minimal, hence also greedy.
\\

$\mu$-minimality is stronger than greediness, and is a powerful and
convenient tool for proving Prop.~\ref{prop-greedy}. However, greediness is
easier to reason with since it only involves local properties of
derivations, while $\mu$-minimality is ``global''. These intuitions are
reflected by, and explain, the following complexity results.
\begin{theorem}
1.\ \emph{Greediness} (deciding whether a given derivation $\pi$ in the
context of a given LGr $G$ is greedy) is in $\L$.
\\  
2.\ \emph{$\mu$-Minimality} (deciding whether it is $\mu$-minimal) is
$\coNP$-complete, even if we restrict to acyclic LGr's.
\end{theorem}
\begin{proof}
1.\  Being leftmost or eager is easily checked in logspace (i.e., is in
$\L$). Checking non-purity can be done by looking for a \emph{last}
inserted useless letter, hence is in $\L$ too.
\\
2.\ $\mu$-minimality is obviously in $\coNP$. Hardness is proved as
Coro.~\ref{coro-mu-min-hard} below, as a byproduct of the reduction we use
for the $\NP$-hardness of Bounded Reachability.
\qed
\end{proof}



\section{Leftist Grammars as Transformers}
\label{sec-ltr}

Some leftist grammars are used as computing devices rather than recognizers of
words. For this purpose, we require a strict separation between input and
output symbols and speak of \emph{leftist transformers}, or shortly
LTr's.

\subsection{Leftist Transformers}
\label{ssec-LTr}

Formally, an LTr is a LGr $G=(\Sigma,P,g)$ where $\Sigma$ is partitioned as
$A\uplus B\uplus C$, and where symbols from $A$ are inactive in $P$ and
are not inserted by $P$ (see Fig.~\ref{fig-type-ltr}). This is denoted
$G:A\vdash C$. 
Here $A$ contains the \emph{input symbols}, $B$ the
\emph{temporary symbols}, and $C$ the
\emph{output symbols}, and $G$ is more conveniently written as
$G=(A,B,C,P,g)$.
When there is no need to distinguish between temporary and
output symbols, we write $G$ under the form $G=(A,D,P,g)$, where $D\egdef
B\cup C$ contains the \emph{``working'' symbols},
%
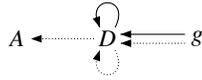
\begin{figure}[htbp]
\centering
{\setlength{\unitlength}{0.60mm}
\begin{picture}(50,10)(0,0)

\gasset{Nframe=n,Nw=6,Nh=6,Nmr=999,loopdiam=5}

\node(s1)(5,5){$A$}
\node(s2)(25,5){$D$}
\node(sg)(45,5){$g$}

\drawedge[syo=1,eyo=1](sg,s2){}
\drawedge[syo=-1,eyo=-1,dash={0.2 0.5}0](sg,s2){}
\drawedge[dash={0.2 0.5}0](s2,s1){}
\drawloop[loopangle=90,loopCW=n](s2){}
\drawloop[loopangle=-90,dash={0.2 0.5}0](s2){}

\end{picture}}
\caption{Type of leftist transformers.}
\label{fig-type-ltr}
\end{figure}

A consequence of the restrictions imposed on LTr's is the following:
\begin{fact}
\label{LTr-basic-invariant}
$A^*D^*$ is an invariant in any LTr $G=(A,D,P,g)$.
\end{fact}

With $G=(A,B,C,P,g)$, we associate a \emph{transformation} (a relation
between words) $R_G\subseteq A^*\times C^*$ defined by
\[
u\mathrel{R_G}v \;\equivdef\; ug\stepgram_G^* vg \:\wedge\: u\in A^* \:\wedge\: v\in C^*
\]
and we say that $G$ \emph{realizes} $R_G$. Finally, a \emph{leftist
transformation} is any relation on words realized by some LTr. By necessity, a
leftist transformation can only relate words written using disjoint
alphabets (this is not contradicted by $\epsilon\mathrel{R_G}\epsilon$).

Leftist transformations respect some structural constraints. In this paper
we shall use the following properties:
\ifthenelse{\boolean{short_version}}{
\begin{proposition}[Closure for leftist transformations]
}{
\begin{proposition}[Closure for leftist transformations, see App.~\ref{app-prop-closure}]
}
\label{prop-closure}
If $G:A\vdash C$ is a leftist transformer, then
$
R_G \: = \: (\invsubword_A.\approx.R_G.\approx)
$.
\end{proposition}

\subsection{Composition}

We say that two leftist transformations $R_1\subseteq A_1^*\times C_1^*$
and $R_2\subseteq A_2^*\times C_2^*$ are \emph{chainable} if $C_1=A_2$ and
$A_1\cap C_2=\emptyset$. Two LTr's are chainable if they realize chainable
transformations.

\begin{theorem}
\label{theo-composition}
The composition $R_1.R_2$ of two chainable leftist transformations is a
leftist transformation.
Furthermore, one can build effectively a linear-sized LTr realizing
$R_1.R_2$ from LTr's realizing $R_1$ and $R_2$.
\end{theorem}
For a proof, assume $G_1 = (A_1,B_1,C_1,P_1,g)$ and $G_2 =
(A_2,B_2,C_2,P_2,g)$ realize $R_1$ and $R_2$. Beyond chainability, we
assume that $A_1\cup B_1$ and $B_2\cup C_2$ are disjoint, which can be
ensured by renaming the intermediary symbols in $B_1$ and $B_2$. 
The composed LTr $G_1.G_2$ is given by 
\[
   G_1 . G_2 \egdef ( A_1, B_1 \cup C_1 \cup B_2, C_2, P_1 \cup P_2, g).
\]
This is indeed a LTr from $A_1$ to $C_2$. See
Fig.~\ref{fig-type-composition} for a schematics of its type. 
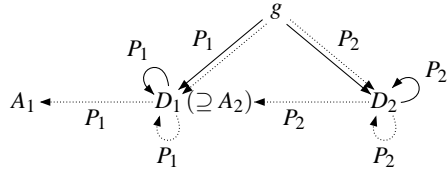
\begin{figure}[htbp]
\centering
{\setlength{\unitlength}{0.60mm}
\begin{picture}(90,30)(0,0)

\gasset{Nframe=n,Nw=7,Nh=7,Nmr=999,loopdiam=5}

\node(s1)(5,10){$A_1$}
\node[Nw=22](s2)(45,10){$D_1\:(\supseteq A_2)$}
\node(s2l)(37,10){} 
\node(s3)(85,10){$D_2$}
\node(sg)(61,30){$g$}

\drawedge[sxo=-1,exo=-1](sg,s3){}
\drawedge[sxo=-1,exo=-1,ELside=r](sg,s2l){{\footnotesize $P_1$}}
\drawedge[sxo=1,exo=1,dash={0.2 0.5}0](sg,s2l){}
\drawedge[sxo=1,exo=1,dash={0.2 0.5}0](sg,s3){{\footnotesize $P_2$}}
\drawedge[dash={0.2 0.5}0](s2l,s1){{\footnotesize $P_1$}}
\drawedge[dash={0.2 0.5}0](s3,s2){{\footnotesize $P_2$}}
\drawloop[loopangle=120,loopCW=n,ELside=r](s2l){{\footnotesize $P_1$}}
\drawloop[loopangle=-90,dash={0.2 0.5}0](s2l){{\footnotesize $P_1$}}
\drawloop[loopangle=30,loopCW=n,ELside=r](s3){{\footnotesize $P_2$}}
\drawloop[loopangle=-90,dash={0.2 0.5}0](s3){{\footnotesize $P_2$}}

\end{picture}}
\caption{The type of $G_1.G_2$.}
\label{fig-type-composition}
\end{figure}
Since $G_1.G_2$ has all rules from $G_1$ and $G_2$ it is clear that
$(\stepgram_{G_1}+ \stepgram_{G_2})\subseteq \stepgram_G$, from which we
deduce $R_{G_1}.R_{G_2}\subseteq R_{G_1.G_2}$. Furthermore, the inclusion in
the other direction also holds:
\ifthenelse{\boolean{short_version}}{
\begin{lemma}[Composition Lemma]
}{
\begin{lemma}[Composition Lemma, see Appendix~\ref{app-lem-composition}]
}
\label{lem-composition}
$R_{G_1 . G_2}= R_{G_1}.R_{G_2}$.
\end{lemma}

\begin{remark}[Associativity]
\label{rem-composition-associative}
The composition $(G_1.G_2).G_3$ is well-defined if and only if
$G_1.(G_2.G_3)$ is. Furthermore, the two expressions denote exactly the
same result.
\qed
\end{remark}




\section{Simple Leftist Transformations}
\label{sec-simple-transformer}

As a tool for Sections~\ref{sec-np-hard} and~\ref{sec-transitive}, 
we now introduce and study restricted
families of leftist grammars (and transformers) where deletion rules are
forbidden (resp., only allowed on $A$).
\\

An \emph{insertion grammar} is a LGr $G=(\Sigma,P,g)$ where $P$ only
contain insertion rules. See Fig.~\ref{fig-type-str} for a graphic
definition. For an arbitrary leftist grammar $G$, we denote with $G^\INS$
the insertion grammar obtained from $G$ by keeping only the insertion
rules.

The \emph{insertion relation} $I_G \subseteq \Sigma^* \times \Sigma^*$
associated with an insertion grammar $G=(\Sigma,P,g)$ is defined by $u\mathrel{I_G}v \equivdef ug \stepgram_G^* vg$. Obviously, $I_G\subseteq\;
\subword_\Sigma$. Observe that $I_G$ is not necessarily a leftist
transformation since it does not require any separation between input and
output symbols.
\\

A \emph{simple} leftist transformer is an LTr $G=(A,B,C,P,g)$ where
 $B=\emptyset$ and where no rule in $P$ erases symbols from $C$.
\begin{figure}[tbp]
\centering
{\setlength{\unitlength}{0.60mm}
\begin{picture}(150,10)(0,0)

\gasset{Nframe=n,Nw=6,Nh=6,Nmr=999,loopdiam=5}

\node(s1)(5,5){$\Sigma$}
\node(sg)(25,5){$g$}

\drawedge[syo=0,eyo=0](sg,s1){}
\drawloop[loopangle=90,loopCW=n](s1){}

\node(ss1)(105,5){$A$}
\node(ss2)(125,5){$C$}
\node(ssg)(145,5){$g$}

\drawedge(ssg,ss2){}
\drawedge[dash={0.2 0.5}0](ss2,ss1){}
\drawloop[loopangle=90,loopCW=n](ss2){}

\end{picture}}
\caption{Types of insertion grammars (left) and simple leftist transformers (right).}
\label{fig-type-str}
\end{figure}
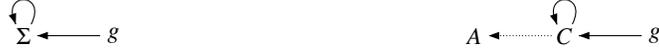
See Fig.~\ref{fig-type-str} for a graphic definition.
We give, without proof, an immediate consequence of the definition:
\begin{lemma}
\label{length-simple-transformer}
Let $G = (A,\emptyset,C,P,g)$ be a simple LTr and assume $ug \stepgram_G^k
vg$ for some $u \in A^*$ and $v \in C^*$. Then $k = \size{u} + \size{v}$.
\end{lemma}

Given a simple LTr $G = (A,\emptyset, C,P,g)$ and two words $u=a_1\cdots
a_n\in A^*$ and $v=c_1\cdots c_m\in C^*$, we say that a non-decreasing map
$h:\{1,\ldots,n\}\rightarrow\{1,\ldots,m\}$ is a \emph{$G$-witness} for $u$
and $v$ if $P$ contains the rules $c_{h(i)} \del a_i$ and $c_{j+1} \ins
c_j$ (for all $i=1,\ldots,n$ and $j=1,\ldots,m$, with the convention that
$c_{m+1}=g$). Finally, we write $u\mathrel{\nabla\!_G}v$ when such a
$G$-witness exists. Clearly, $\nabla\!_G\subseteq R_G$. Indeed, when $G$ is
a simple transformer, $\nabla\!_G$ can be used as a restricted version of
$R_G$ that is easier to control and reason about.


\ifthenelse{\boolean{short_version}}{
\begin{lemma}
}{
\begin{lemma}[See App.~\ref{app-lemma-nabla}]
}
\label{lemma-nabla}
Let $G = (A,\emptyset,C,P,g)$ be a simple LTr.
Then $R_G = \nabla\!_G.I_{G^\INS}$.
\end{lemma}
Combining Lemma~\ref{lemma-nabla} with 
$\Id_{C^*}\subseteq I_{G^\INS}\subseteq\;\subword_C$, 
we obtain the  following weaker but simpler statement.
\begin{corollary}
\label{coro-nabla}
Let $G=(A,\emptyset,C,P,g)$ be a simple LTr. Then $\nabla\!_G \subseteq R_G \subseteq \nabla\!_G.\subword_C$.
\end{corollary}

\subsection{Union of Simple Leftist Transformers}

We now consider the combination of two simple LTr's $G_1 =
(A,\emptyset,C_1,P_1,g)$ and $G_2 = (A,\emptyset,C_2,P_2,g)$ that transform
from a same $A$ to disjoint output alphabets, i.e., with $C_1 \cap C_2 =
\emptyset$. We define their \emph{union} with $G_1 + G_2 \:\egdef\:
(A,\emptyset,C_1\cup C_2,P_1\cup P_2,g)$. This is clearly a simple LTr with
$(R_{G_1}+R_{G_2})\subseteq R_{G_1+G_2}$.
It
further satisfies:
\begin{lemma}
\label{union-simple-transformer}
If $u\mathrel{R_{G_1 + G_2}}v$ then $u\mathrel{(R_{G_1} + R_{G_2})}v'$ for some $v'
\subword v$.
\end{lemma}
\begin{proof}
Assume $u\mathrel{R_{G_1 + G_2}}v$. With Cor.~\ref{coro-nabla}, we obtain
$u\mathrel{\nabla\!_{G_1 + G_2}} v'$ for some $v'=c_1\cdots c_m\subword v$.
Hence $G_1 + G_2$ has insertion rules $c_{j+1}\ins c_j$ for all
$j=1,\ldots,m$, and deletion rules of the form $c_{h(i)}\del u[i]$. Since
$C_1$ and $C_2$ are disjoint, either all these rules are in $G_1$ (and
$u\mathrel{\nabla\!_{G_1}}v'$), or they are all in $G_2$ (and
$u\mathrel{\nabla\!_{G_2}}v'$). Hence $u\mathrel{(R_{G_1} + R_{G_2})}v'$.
\qed
\end{proof}


\section{Encoding \3SAT with Acyclic Leftist Transformers}
\label{sec-np-hard}

This section proves the following result.
\begin{theorem}
\label{theo-bounded-reachibility-npcomp}
\emph{Bounded Reachability} and \emph{Exact Bounded Reachability} in
leftist grammars are $\NP$-complete, even when restricting to acyclic
grammars.
\end{theorem}

(Exact) Bounded Reachability is the question whether there exists a
$n$-step derivation $u\stepgram^n v$ (respectively, a derivation
$u\stepgram^{\leq n}v$ of non-exact length at most $n$) between given $u$
and $v$. These questions are among the simplest reachability questions and,
since we consider that the input $n$ is given in unary,\footnote{It is
  natural to begin with this assumption when considering fundamental
  aspects of reachability since writing $n$ more succinctly would blur the
  complexity-theoretical picture.} they are obviously in $\NP$ for leftist
grammars (and all semi-Thue systems).

Consequently, our contribution in this paper is the $\NP$-hardness part.
This is proved by encoding $\3SAT$ instances in leftist grammars where
reaching a given final $v$ amounts to guessing a valuation that satisfies
the formula. While the idea of the reduction is easy to grasp, the
technicalities involved are heavy and it would be difficult to really prove
the correctness of the reduction without relying on a compositional
framework like the one we develop in this paper. It is indeed very tempting
to ``prove'' it by just running an example.

Rather than adopting this easy way, we shall describe the reduction as a
composition of simple leftist transformers and use our composition theorems
to break down the correctness proof in smaller, manageable parts. Once the
ideas underlying the reduction are grasped, a good deal of the reasoning is
of the type-checking kind: verifying that the conditions required for
composing transformers are met.

Throughout this section we assume a generic $\3SAT$ instance $\Phi =
\bigwedge_{i=1}^m C_i$ with $m$ 3-clauses on $n$ Boolean variables in
$X=\{x_1,\ldots,x_n\}$. Each clause has the form $C_i = \bigvee_{k=1}^3
\epsilon_{i,k} x_{i,k}$ for some polarity $\epsilon_{i,k} \in \{ +, - \}$
and $x_{i,k} \in X$. (There are two additional assumptions on
  $\Phi$ that we postpone until the proof of Coro.~\ref{coro-np-red-1} for
  clarity.) We use standard model-theoretical notation like $\sat\Phi$
(validity), or $\sigma\sat\Phi$ (entailment) when $\sigma$ is a Boolean
formula or a Boolean valuation of some variables. 

We write $\sigma[x\mapsto b]$ for the extension of a valuation $\sigma$
with $(x,b)$, assuming $x\not\in\Dom(\sigma)$. Finally, for a valuation
$\theta:X\rightarrow \{\top,\bot\}$ and some $j=0,\ldots,n$, we write
$\theta_j$ to denote the restriction $\theta_{|\{x_1,\ldots,x_j\}}$ of
$\theta$ on the first $j$ variables.

\subsection{Associating an LTr $G_\Phi$ with $\Phi$}

For the encoding, we use an alphabet
$\Sigma=\{T_i^j,U_i^j,{T'}_i^j,{U'}_i^j~|~i=1,\ldots,m \:\wedge\: j =
0,\ldots,n\}$, i.e., $4(n+1)$ symbols for each clause. The choice of the
symbols is that a $U$ means ``\emph{Undetermined}'' and a $T$ means
``\emph{True}'', or determined to be valid.

For $j=0,\ldots,n$, let $V_j\egdef\{U_1^j,\ldots,U_m^j,T_1^j,\ldots,T_m^j\}$,
$V'_j\egdef\{{U'}_1^j,\ldots,{U'}_m^j,{T'}_1^j,\ldots,{T'}_m^j\}$, and
$W_j\egdef V_j\cup V'_j$, so that
$\Sigma$ is partitioned in levels with $\Sigma=\bigcup_{j=0}^{n} W_j$.
With each $x_j\in X$ we associate two intermediary LTr's:
\begin{xalignat*}{2}
       G^\top_j &\egdef (W_{j-1}, \emptyset, V_j, P_j, g),
&
       G^\bot_j &\egdef (W_{j-1}, \emptyset, V'_j, P'_j, g)
\end{xalignat*}
with sets of rules $P_j$ and $P'_j$. The rules for $G^\top_j$ are given in
Fig.~\ref{fig-np-rules}: some deletion rules are conditional, depending on
whether $x_j$ appears in the clauses $C_1,\ldots,C_m$. The rules for
$G^\bot_j$ are obtained by switching primed  and unprimed symbols, and by
having conditional rules based on whether $\neg x_j$ appears in the $C_i$'s.
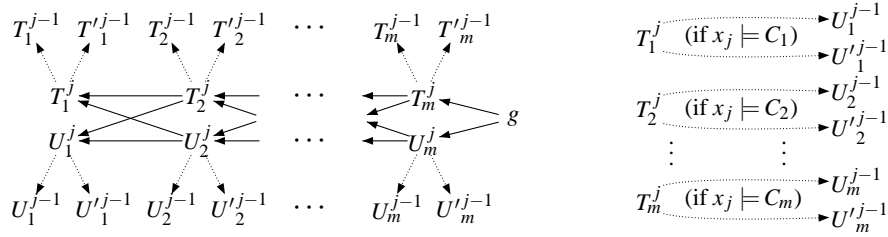
\begin{figure}[htbp]
\centering
{\setlength{\unitlength}{0.60mm}
\begin{picture}(190,40)(0,0)

\def\angfor{18.5}
\def\aol{198.5} 
\def\aou{161.5} 

\gasset{Nframe=n,Nw=7,Nh=7,Nmr=999,loopdiam=5}

\node(t1)(10,25){$T_1^j$}
\node(u1)(10,15){$U_1^j$}
\node(t2)(40,25){$T_2^j$}
\node(u2)(40,15){$U_2^j$}
\node(tdots)(65,25){\large $\cdots$}
\node(udots)(65,15){\large $\cdots$}
\node(tm)(90,25){$T_m^j$}
\node(um)(90,15){$U_m^j$}
\node(sg)(110,20){$g$}
\drawedge(sg,tm){}
\drawedge(sg,um){}
\drawedge(t2,t1){}
\drawedge(t2,u1){}
\drawedge(u2,t1){}
\drawedge(u2,u1){}

{\gasset{ilength=10,flength=10}
\imark[iangle=0](t2)	\imark[iangle=-\angfor](t2)
\imark[iangle=0](u2)	\imark[iangle=\angfor](u2)
\fmark[fangle=180](tm)	\fmark[fangle=\aol](tm)
\fmark[fangle=180](um)	\fmark[fangle=\aou](um)
}

\node(tt1)(2.5,40){$T_1^{j-1\!\!\!\!\!\!\!\!}$}
\node(tt1p)(17.5,40){${T'}_1^{j-1\!\!\!\!\!\!\!\!}$}
\node(tt2)(32.5,40){$T_2^{j-1\!\!\!\!\!\!\!\!}$}
\node(tt2p)(47.5,40){${T'}_2^{j-1\!\!\!\!\!\!\!\!}$}
\node(ttdots)(65,40){\large $\cdots$}
\node(ttm)(82.5,40){$T_m^{j-1\!\!\!\!\!\!\!\!}$}
\node(ttmp)(97.5,40){${T'}_m^{j-1\!\!\!\!\!\!\!\!}$}
\node(uu1)(2.5,0){$U_1^{j-1\!\!\!\!\!\!\!\!}$}
\node(uu1p)(17.5,0){${U'}_1^{j-1\!\!\!\!\!\!\!\!}$}
\node(uu2)(32.5,0){$U_2^{j-1\!\!\!\!\!\!\!\!}$}
\node(uu2p)(47.5,0){${U'}_2^{j-1\!\!\!\!\!\!\!\!}$}
\node(uudots)(65,0){\large $\cdots$}
\node(uum)(82.5,0){$U_m^{j-1\!\!\!\!\!\!\!\!}$}
\node(uump)(97.5,0){${U'}_m^{j-1\!\!\!\!\!\!\!\!}$}

{\gasset{dash={0.2 0.5}0}
\drawedge(t1,tt1){}\drawedge(t1,tt1p){}
\drawedge(t2,tt2){}\drawedge(t2,tt2p){}
\drawedge(tm,ttm){}\drawedge(tm,ttmp){}
\drawedge(u1,uu1){}\drawedge(u1,uu1p){}
\drawedge(u2,uu2){}\drawedge(u2,uu2p){}
\drawedge(um,uum){}\drawedge(um,uump){}
}

{\gasset{dash={0.2 0.5}0,Nframe=n,Nw=10}
\node(t1)(140,38){$T_1^j$}
\node(t2)(140,22){$T_2^j$}
\node(tdots)(145,14){$\vdots$}
\node(tdots)(170,14){$\vdots$}
\node(tm)(140,2){$T_m^j$}

{\gasset{Nw=0,Nframe=n,ExtNL=y,NLdist=0.5,NLangle=0} 
\node(u1)(180,42){$U_1^{j-1\!\!\!\!\!\!\!\!}$}
\node(up1)(180,34){${U'}_1^{j-1\!\!\!\!\!\!\!\!}$}
\node(u2)(180,26){$U_2^{j-1\!\!\!\!\!\!\!\!}$}
\node(up2)(180,18){${U'}_2^{j-1\!\!\!\!\!\!\!\!}$}
\node(um)(180,6){$U_m^{j-1\!\!\!\!\!\!\!\!}$}
\node(upm)(180,-2){${U'}_m^{j-1\!\!\!\!\!\!\!\!}$}
}

{\gasset{ExtNL=y,NLangle=0,NLdist=3}
\nodelabel(t1){(if $x_j\sat C_1$)}
\nodelabel(t2){(if $x_j\sat C_2$)}
\nodelabel(tm){(if $x_j\sat C_m$)}
\drawbpedge(t1,60,2,u1,178,45){}
\drawbpedge(t2,60,2,u2,178,45){}
\drawbpedge(tm,60,2,um,178,45){}
\drawbpedge(t1,-60,2,up1,182,45){}
\drawbpedge(t2,-60,2,up2,182,45){}
\drawbpedge(tm,-60,2,upm,182,45){}
}

}

\end{picture}}
\caption{$P_j$, the rules for $G_J^\top$: Fixed part on left, conditional
  part on right.}
\label{fig-np-rules}
\end{figure}
One easily checks that $G_j^\top$ and $G_j^\bot$ are indeed simple transformers. They have same inputs and disjoint outputs so that the union
$(G^\top_j+G^\bot_j):W_{j-1}\vdash W_j$ is well-defined. Hence the
following composition is well-formed:
\[
   G_\Phi\egdef (G^\top_1 + G^\bot_1). (G^\top_2 + G^\bot_2) \cdots (G^\top_n + G^\bot_n).
\]
We conclude the definition of $G_\Phi$ with an intuitive explanation of the
idea behind the reduction. $G_\Phi$ operates on the word $u_0=U_1^0\cdots
U_m^0$ where each $U_i^0$ stands for ``\emph{the validity of clause $C_i$
  is undetermined at step $0$ (i.e., at the beginning)}''. At step $j$,
$G^\top_j+G^\bot_j$ picks a valuation for $x_j$: $G^\top_j$ picks
``$x_j=\top$'' while $G^\bot_j$ picks ``$x_j=\bot$''. This transforms
$U_i^{j-1}$ into $U_i^j$, and $T_i^{j-1}$ into $T_i^j$, moving them to the
next level. Furthermore, an undetermined $U_i^{j-1}$ can be transformed
into $T_i^j$ if $C_i$ is satisfied by $x_j$. In addition, and because
$G^\top_j$ and $G^\bot_j$ must have disjoint output alphabets, the symbols
in the $V_j$'s come in two copies (hence the $V'_j$'s) that behave
identically when they are input in the transformer for the next step.

The reduction is concluded with the following claim that we prove by
combining Corollaries~\ref{coro-np-red-1} and~\ref{coro-np-red-2} below.
\begin{align}
\notag
\Phi \text{ is satisfiable} \text{ iff } 
	U_1^ 0U_2^0 \cdots U_m^0 g & \stepgram^{2mn}_{G_\Phi}T_1^n T_2^n \cdots T_m^n g\\[.3em]
\tag{Correctness}
\text{ iff } 
	U_1^ 0U_2^0 \cdots U_m^0 g & \stepgram^{\leq 2mn}_{G_\Phi}T_1^n T_2^n \cdots T_m^n g\\[.3em]
\notag
\text{ iff } 
	U_1^ 0U_2^0 \cdots U_m^0 g & \stepgram^{*}_{G_\Phi}T_1^n T_2^n \cdots T_m^n g.
\end{align}

Observe finally that $G_\Phi$ is an acyclic grammar in the sense
of~\cite{jurdzinski2007c}, that is to say, its rules define an acyclic
``\emph{may-act-upon}'' relation between symbols. Such grammars are much
weaker than general LGr's since, e.g., languages recognized by LGr's with
acyclic deletion rules (and arbitrary insertion rules) are
regular~\cite{jurdzinski2007c}.
\begin{remark}
The construction of $G_\Phi$ from $\Phi$, mostly amounting to
copying operations for the $G_j^\top$'s and $G_j^\bot$'s, to type-checking
and sets-joining operations for the composition of the LTr's, can be
carried out in logarithmic space.
\qed
\end{remark}

\subsection{Correctness of the Reduction}

We say that a word $u$ is \emph{$j$-clean} if it has exactly $m$ symbols
and if $u[i]\in\{T_i^j,{T'}_i^j,U_i^j,{U'}_i^j\}$ for all $i=1,\ldots,m$.
It is \emph{$\top$-homogeneous} (resp.\
\emph{$\bot$-homogeneous}) if it does not contain any (resp., only
contains) primed symbols.

Let $0\leq j\leq n$ and $\theta_j$ be a Boolean valuation of
$x_1,\ldots,x_j$: we say that a $j$-clean $u$ \emph{respects} ($\Phi$
under) $\theta_j$ when, for all $i=1,\ldots,m$, $\theta_j\sat C_i$ when
$u[i]$ is determined (i.e., $\in T_i^j+{T'}_i^j$). Finally $u$ \emph{codes}
($\Phi$ under) $\theta_j$ if additionally each $u[i]$ is determined when
$\theta_j\sat C_i$. Thus, a word $u$ that codes some $\theta_j$ exactly
lists (via determined symbols) the clauses of $\Phi$ made valid by
$\theta_j$, and the only flexibility in $u$ is in using the primed or the
unprimed copy of the symbols. Hence there is only one $j$-clean
$u$ coding $\theta_j$ that is $\top$-homogeneous, and only one that is $\bot$-homogeneous. If $u$ respects $\theta_j$ instead of coding it, more latitude exists
since symbols may be undetermined even if the corresponding clause is valid
under $\theta_j$.
\\

Assume that, for some $j\in\{1,\ldots,n\}$, $u_{j-1}$ codes $\theta_{j-1}$
and $u_{j}$ codes $\theta_{j}$. Write $b$ for $\theta(x_{j})$ (NB: $b
\in\{\top,\bot\}$).
\begin{lemma}
\label{lem-nabla-one-step}
If $u_{j}$ is $b$-homogeneous 
then $u_{j-1} \mathrel{\nabla\!_{G^b_j}} u_{j}$.
\end{lemma}
\begin{proof}
Let $h\egdef \Id_{\{1,\ldots,m\}}$. We claim that $h$ is a $G_j^b$-witness for
$u_{j-1}$ and $u_j$, i.e., that $G^b_j$ contains the required insertion and
deletion rules.
\\
\textbf{Insertions.} $G^b_j$ has all insertion rules $g\ins u_{j}[m]\ins
u_{j}[m-1]\ins \ldots \ins u_{j}[1]$ (leftmost rules in
Fig.~\ref{fig-np-rules}) since $u_{j}$ is $b$-homogeneous.
\\  
\textbf{Deletions.} $G_j^b$ has all deletion rules $u_{j}[i]\del
u_{j-1}[i]$. Firstly, both undetermined symbols $U^i_{j}$ and ${U'}^i_{j}$
may delete their counterparts $U^i_{j-1}$ and ${U'}^i_{j-1}$, and similarly
for the determined symbols (the unconditional deletion rules in
Fig.~\ref{fig-np-rules}). This is used if $C_i$ is not more valid under
$\theta_{j}$ than under $\theta_{j-1}$. Secondly, if $C_i$ is valid under
$\theta_{j}$ but not under $\theta_{j-1}$, then $x_{j}\sat C_i$ (or $\neg
x_{j}\sat C_i$, depending on $b$) and the conditional rules in
Fig.~\ref{fig-np-rules} allow a determined $T_i^{j}$ (or ${T'}_i^{j}$
depending on $b$) to delete $U_i^{j-1}$ or ${U'}_i^{j-1}$.
\qed
\end{proof}
\begin{lemma}
\label{rules-one-step}
If $u_{j}$ is $b$-homogeneous, then $u_{j-1}g \stepgram_{G^b_{j}}^{2m}
u_{j}g$.
\end{lemma}
\begin{proof}
From $u_{j-1}\mathrel{\nabla\!_{G^b_j}}u_{j}$
(Lemma~\ref{lem-nabla-one-step}) we deduce $u_{j-1}\mathrel{R_{G^b_j}}u_{j}$,
i.e., $u_{j-1}g\stepgram^*_{G^b_j}u_{j}g$, by Lemma~\ref{lemma-nabla}, and
then $u_{j-1}g\stepgram^{2m}_{G^b_j}u_{j}g$ by
Lemma~\ref{length-simple-transformer}.
\qed
\end{proof}

\begin{corollary}
\label{coro-np-red-1}
If  $\Phi$ is satisfiable, then
  $U_1^0 \cdots U_m^0g \stepgram_{G_\Phi}^{2mn} T_1^n \cdots T_m^ng$.
\end{corollary}
\begin{proof}
Since $\Phi$ is satisfiable, $\theta\sat\Phi$ for some valuation $\theta$.
For $j=1,\ldots,m$, we write $b_j$ for $\theta(x_j)$ and let $u_j$ be the only
$j$-clean $b_j$-homogeneous word that codes for $\theta_j$.

We now make two assumptions on $\Phi$ that are no loss of generality. First
we require that no clause $C_i$ contains both a literal and its negation,
hence no $C_i$ is tautologically valid. Then $u_0\egdef U_1^0 \cdots U_m^0$
codes the empty valuation $\theta_0$. Second, we require that $\Phi$ is
only satisfiable with $b_n=\top$ (which can be easily ensured by adding a
few extra variables). Then necessarily $u_n=T_1^n \cdots T_m^n$.

Lemma~\ref{rules-one-step} gives
$u_0g \stepgram_{G_1^{b_1}}^{2m} u_1g  \stepgram_{G_2^{b_2}}^{2m} u_{2}g
\cdots  \stepgram_{G_n^{b_n}}^{2m}u_ng$.
Since $\stepgram_{G_j^b}\subseteq
\stepgram_{G_j}\subseteq\stepgram_{G_\Phi}$ for all $b$ and $j$, we deduce
$u_0g\stepgram_{G_\Phi}^{2mn} u_ng$ as claimed.
\qed
\end{proof}

Fix some $\theta$, some $j\in\{1,\ldots,n\}$ and let $b=\theta(x_j)$.
\begin{lemma}
\label{nabla-phi}
If $u$ respects $\theta_{j-1}$ and
 $u\mathrel{\nabla\!_{G^b_j}}v$,
then $v$ respects $\theta_j$.
\end{lemma}
\begin{proof}
Write $l$ for $\size{v}$. From $u\mathrel{\nabla\!_{G^b_j}}v$ (witnessed by
some $h$) we deduce that $G^b_j$ has insertion rules $g\ins v[l]\ins v[l-1]
\ins\ldots\ins v[1]$. Inspecting Fig.~\ref{fig-np-rules}, we conclude that
necessarily $l\leq m$. Since deletion rules $v[h(i)]\del u[i]$ are required
for all $i=1,\ldots,m$, we further see from Fig.~\ref{fig-np-rules} that
$h$ is injective, so that $l\geq m$. Finally $l=m$,
$h=\Id_{\{1,\ldots,m\}}$, $v$ is $j$-clean and $b$-homogeneous.

Now, knowing that $G_j^b$ contains the rules $v[i]\del u[i]$, we show that
$v$ respects $\theta_j$. Suppose, by way of contradiction, that it does
not. Thus there is some $i\in\{1,\ldots,m\}$ with $v[i]=T_i^j$ (assuming
$b=\top$ w.l.o.g.) while $\theta_j\not\sat C_i$ (so that
$\theta_{j-1}\not\sat C_i$). From $\theta_j\not\sat C_i$ we deduce that
$x_j\not\sat C_i$. Hence $G^b_j$ does not have the conditional rules
$T_i^j\del U_i^{j-1}$ and $T_i^j\del {U'}_i^{j-1}$. Thus
$u[i]\not\in\{U_i^{j-1},{U'}_i^{j-1}\}$. But then $u$ does not respect
$\theta_{j-1}$, contradicting our assumption.
\qed
\end{proof}
We immediately deduce:
\begin{lemma}
\label{subword-phi}
If $x\mathrel{R_{G^b_j}}y$ and there is some $u\subword x$ that respects
$\theta_{j-1}$, then there is some $v\subword y$ that respects $\theta_j$.
\end{lemma}
\begin{proof}
From the Closure Property~\ref{prop-closure}, we get
$u\mathrel{R_{G^b_j}}y$. Then, from $R_{G_j^b}\subseteq
\nabla\!_{G_j^b}.\subword$ (Coro.~\ref{coro-nabla}) we deduce
$u\mathrel{\nabla\!_{G_j^b}} v$ for some $v\subword y$. Now $v$ respects
$\theta_j$ thanks to Lemma~\ref{nabla-phi}.
\qed
\end{proof}

\begin{corollary}
\label{coro-np-red-2}
If  $U_1^0 \cdots U_m^0 g\stepgram^*_{G_\Phi} T_1^n \cdots T_m^n g$, then 
 $\Phi$ is satisfiable.
\end{corollary}
\begin{proof}
Write $u_0$ for $U_1^0 \cdots U_m^0$ and $u_n$ for $T_1^n \cdots T_m^n$.
From the definition of $G_\Phi$ and the Composition Lemma~\ref{lem-composition},
we deduce that there exist some words $u_1,\ldots,u_{n-1}$ such that
$u_{j-1}\mathrel{R_{G_j^\top+G_j^\bot}} u_j$ for all $j=1,\ldots,n$.

With Lemma~\ref{union-simple-transformer}, we further deduce that there
exist some words $u'_1, \ldots, u'_n$ and Boolean values $b_1,\ldots,b_n$
such that $u'_j \subword u_j$ and $u_{j-1}\mathrel{R_{G^{b_j}_j}}u'_{j}$
for all $j=1,\ldots,n$. Hence also $u'_{j-1}\mathrel{R_{G^{b_j}_j}}u'_{j}$
by Prop.~\ref{prop-closure} (and letting $u'_0=u_0$).

Write $\theta$ for $[x_1\mapsto b_1,\ldots,x_n\mapsto b_n]$. With
Lemma~\ref{subword-phi}, induction on $j$, and since $u'_0$ respects
$\theta_0$, we further deduce that there exists some words
$u''_1,\ldots,u''_n$ such that, for all $j=1,\ldots,n$, $u''_j\subword
u'_j$ and $u''_j$ respects $\theta_j$. From $\size{u''_n}=m$ (it respects
$\theta$) and $u''_n\subword u_n$, we deduce that $u''_n=u_n$. Finally,
$\theta\sat\Phi$
since $u''_n$ respects $\theta$
and $u''_n=u_n=T_1^n\cdots T_m^n$.
\qed
\end{proof}
\begin{corollary}
\label{coro-mu-min-hard}
$\mu$-Minimality of a derivation is $\coNP$-hard.
\end{corollary}
\begin{proof}[Sketch]
We 
define $G'_\Phi$ 
by taking $G_\Phi$,
adding $k$ extra symbols
$a_1,\ldots,a_k$, and adding the following two sets of rules:\\
(1) all $a_{i-1}\ins a_{i}$ and $a_{i-1}\del a_{i}$ for $i=1,\ldots,k$
(with the convention that $a_0$ is $T^n_1$);
\\
(2) all $a_k\del U_i^0$ for $i=1,\ldots,m$.

Observe that $G'_\Phi$ is acyclic. It has a derivation $\pi:U_1^0\cdots
U_m^0g\stepgram^{2m+2k}T_1^n\cdots T_m^n g$ of the following form:
\begin{align*}
U_1^0\cdots U_m^0g
 \stepgram^m
U_1^0\cdots U_m^0 T_1^n\cdots T_m^n g
& \stepgram^k
U_1^0\cdots U_m^0 a_ka_{k-1}\cdots a_1 T_1^n\cdots T_m^n g
\\
&\stepgram^m
a_ka_{k-1}\cdots a_1 T_1^n\cdots T_m^n g
\stepgram^k
T_1^n\cdots T_m^n g.
\end{align*}
This derivation uses the extra symbols to bypass the normal behaviour of
$G_\Phi$. If $k$ is large enough, i.e., $k>m(n-1)$, $\pi$ is $\mu$-minimal
if, and only if, $\Phi$ is not satisfiable.
\qed
\end{proof}


\section{Anchored Leftist Transformers and Their Transitive Closure}
\label{sec-transitive}

When $b_1,b_2 \in B$ are two different working symbols, and $(A,B,C,P,g)$
is a LTr, we call $G = (A,B,C,b_1,b_2,P,g)$ an \emph{anchored} LTr, or
shorly an ALTr.
With an ALTr $G$ we associate
an \emph{anchored tranformation} $S_G \subseteq A^* \times C^*$ defined by
\[
       u\mathrel{S_G}v \:\equivdef\: b_1 u g \stepgram^*_G b_2 v g.
\]
Here the \emph{anchors} $b_1,b_2$ are used to control what happens at the
left-hand end of transformed words. Mostly, they ensure that the derivation
$b_1ug\stepgram^* b_2vg$ goes all the way to the left and erases $b_1$
rather than stopping earlier. One intuitive way of seeing $S_G$ is that it
is a variant of $R_G$ \emph{restricted to derivations that replace the
  anchors}.


A first difficulty for building the transitive closure of an anchored
transformation $S_G\subseteq A^*\times C^*$ is that the input and output
sets are disjoint (a requirement that allowed the developments of
Sections~\ref{sec-ltr} and~\ref{sec-simple-transformer}). To circumvent
this, we assume w.l.o.g.\  that $A$ and $C$ are two different copies of a
same set, equipped with a bijective renaming $\bar{h}:C^*\rightarrow A^*$.
Then, the closure $S_G.(\bar{h}.S_G)^*$ behaves like we would want $S_G^+$
to behave.

For the rest of this section, we assume $h$ is a bijection
between $C$ and $A$. W.l.o.g., we write $A$ and $C$ under the forms
$A=\{a_1,\ldots,a_n\}$ and $C=\{c_1,\ldots,c_n\}$ so that $h(c_i)=a_i$ for
all $i=1,\ldots,n$. Then $h$ is lifted as a (bijective) morphism
$\bar{h}:C^*\rightarrow A^*$ that we sometimes see as a relation between
words.

The exact statement we prove in this section is the following:
\begin{theorem}[Transitive Closure]
\label{theo-transitive}
Let $G:A\vdash C$ be an ALTr such that $S_G=S_G.\subword_C$. Then there
exists an ALTr $G^{(+)}:A\vdash C$ such that $S_{G^{(+)}} =
S_G.(\bar{h}.S_G)^*$.
\\
Furthermore, it is possible to build $G^{(+)}$ from $G$ using only
logarithmic space.
\end{theorem}

Let $b_1,b_2\not\in A\cup C$. The ALTr $\renaming_{b_2,b_1} \egdef
(C,b_2,b_1,A,P_\renaming,g)$ with 
\[
P_\renaming
\:\egdef\:
\left\{
\left.
\begin{array}{c}
g\ins a_i, 
a_i\ins a_j, 
a_i\ins b_1
\\
a_i\del c_i,
b_1\del b_2
\end{array}
\right|
\text{for all $i,j=1,\ldots,n$}
\right\}
\]
is called a \emph{renamer (of $C$ to $A$)}, and often
shortly written $\renaming$.
Observe that $\renaming:C\vdash A$ is indeed an	 ALTr. It further satisfies
$
	     S_{\renaming} \: = \;\: \approx.\subword.\bar{h}.
$

We shall now glue an ALTr $G:A\vdash C$ with the renamer $\renaming:C\vdash
A$ into some larger LGr $H$. But before this can be done we need to put
some wrapping control on $G$ (and on $\renaming$) that will let us track
what comes from $G$ inside $H$'s derivations.

Formally, given an ALTr $G = (A,B,C,b_1,b_2,P,g)$ and two new anchor
symbols $\B_1,\B_2\not\in\Sigma_g$, we let $\Sigma_\B\egdef\{\B_1,\B_2\}$ and
define a new ALTr  $F_{G,\B_1,\B_2}$ (or shortly just $F_G$) for
``\emph{wrapping $G$ with $\B_1,\B_2$}'', and given by
$
  F_{G,\B_1,\B_2} \egdef
  (\overline{A},\overline{B},\overline{C},\B_1,\B_2,P',g)
$
where
\\
-- $\overline{A} \egdef A \cup A' \cup \{b_1,b_1'\}$, $A',b'_1$ being a copy of $A,b_1$,
\\
-- $\overline{B} \egdef \{\B_1,\B_2\} \cup B \setminus \{b_1\}$,
\\
-- $\overline{C} \egdef C \cup \{b_2\} \cup C' \cup B' \setminus \{b_1'\},$ $B'$ and $C'$ being copies of $B$ and $C$.
\\
Finally, let $D \egdef C \cup B$ and $D' \egdef C' \cup B'$.
(The copies are denoted by priming the original symbols, and a
primed set like $A'=\{a'~|~ a\in A\}$ is just the set of corresponding
primed symbols.)
The rules in $P'$ are derived from the rules of $P$ in the following way.
\ifthenelse{\boolean{short_version}}{
}{
(See Fig.~\ref{fig-type-F} in App.~\ref{app-transitive} for a schematic type.)
}
  \begin{description}
  \item[kept:]
    $P'$ retains all rules of $P$ that do not erase a letter in $A \cup \{b_1\}$,
  \item[replace:]
    $P'$ has a rule $d'\del a$ for each rule $d\del a$ in $P$ that erases a
    letter in $A \cup \{b_1\}$,
  \item[mirror:]
    $P'$ has a rule $d\ins d'$ for each $d\in D$,
  \item[clean:]
    $P'$ has all rules $d'\del e'$  and $\B_2 \del a'$ for  $d',e'\in D' \setminus \{b_1'\}$
    and	 $a'\in A' \cup \{b_1'\}$,
  \item[$b$-rules:]
    $P'$ has the rules $\B_2\del \B_1$ and all rules $d'\ins \B_2$ for $d'\in D' \setminus \{b_1'\}$.
  \end{description}

We now relate the derivations in $G$ and the derivations in
$F_G$. For this, assume $u \in (A + b_1)^*$ and $v \in (C +
b_2)^+$.
\ifthenelse{\boolean{short_version}}{
\begin{lemma}
}{
\begin{lemma}[See App.~\ref{app-lem-tr1}]
}
\label{lem-tr1}
1.\  If $u.g \stepgram^+_{G} v.g$ then for all words $\alpha \in (A' +
b_1')^*$ there exists a symbol $\beta \in C' \cup \{b_2'\}$ such that
$\B_1.\alpha.u.g \stepgram^+_{F_{G}} \B_1.\alpha.\beta.v.g
\stepgram^+_{F_{G}} \B_2.\beta.v.g$.
\\
2.\ Reciprocally,
 for all $\alpha \in (A' + b_1')^*$, for all $\beta \in (C' + b_2')^+$
  if $\B_1.\alpha.u.g \stepgram^+_{F_{G}} \B_2.\beta.v.g$ then
$u.g \stepgram^+_{G} v.g$.
\end{lemma}
Thus we can relate anchored derivations in $F_{G}$ with
anchored derivations in $G$ via:
\begin{corollary}
\label{lem-fg=g}
Let $u \in (A + b_1)^*$ and $v \in (C + b_2)^+$. Then $b_1.u.g
\stepgram_{G}^+ b_2.v.g$ if and only if there exists $\beta \in (C' \cup
\{b_2'\})$ such that $\B_1.\alpha.b_1.u.g \stepgram^+_{F_{G}}
\B_2.\beta.b_2.v.g$. In other words, $u \mathrel{S_G} v$ iff $\alpha.b_1.u
\mathrel{S_{F_{G}}} \beta.b_2.v$ for some $\beta \in (C' \cup \{b_2'\})$.
\end{corollary}

We may now glue the wrapped versions of $G$ and its associated $\renaming$.
Recall that $F_{G} =
(\overline{A},\overline{B},\overline{C},\B_1,\B_2,P',g)$. We denote the set
of new symbols with $\Sigma\egdef \overline{A} \cup \overline{B} \cup
\overline{C}$ and observe that $F_\renaming$ (short for
$F_{\renaming_{b_2,b_1},\B_2,\B_1}$), being some $(C \cup C' \cup
\{b_2,b'_2\},\Sigma_\B,\overline{A},\B_2,\B_1,P'_\renaming,g)$, does not
use more symbols.
Let $H\egdef (\Sigma,P_H,g)$ be the LGr such that and $P_H = P'
\cup P'_\renaming$. Essentially, $H$ is a union of the two wrapping ALTr's.
\ifthenelse{\boolean{short_version}}{
}{
(See Fig.~\ref{fig-type-H} in App.~\ref{ssec-app-H} for a schematic
description). 
}
Note that $H$ is \emph{not a LTr} since it does not respect
any distinction between input, intermediary, and output symbols.

\ifthenelse{\boolean{short_version}}{
\begin{lemma}
}{
\begin{lemma}[See App.~\ref{app-lem-H}]
}
\label{lem-H}
Let $\alpha, \beta \in A'^+$ and $u,v \in A^*$. If $\B_1.\alpha.u.g
\stepgram^*_{H} \B_1.\beta.v.g$ and $S_G = (\subword_A.S_G.\subword_C)$
then $u \mathrel{\subword_A.(S_G.\bar{h})^*} v$.
\end{lemma}

We now extend $H$ to turn it into an ALTr
$H':\dot{A}\vdash A\cup A'$, introducing again new copies, denoted
$\dot{a}$, \ldots, of previously used symbols and writing $\dot{u} = \dot{a_1} \dot{a_2} \ldots \dot{a_n}$ for the dotted copy of
some $u = a_1 a_2 \ldots a_n$.
Formally,
\[
H'\egdef (\dot{A},B \cup B' \cup C
  \cup C' \cup \{\B_1,\B_2,\dot{\B_1},\dot{\B_2}\},A \cup A',
  \dot{\B_1},\dot{\B_2},P'',g)
\]
where $P''$ extends $P_H$ by the rules $\dot{\B_2} \del
\dot{\B_1}$, $\B_1 \ins \dot{\B_2}$, and  all $a \del \dot{a}$ for $a \in
A$.

The anchored transformation $S_{H'}$ computed by $H'$ is captured by the following:
\ifthenelse{\boolean{short_version}}{
\begin{lemma}
}{
\begin{lemma}[See App.~\ref{app-lem-H'}]
}
\label{lem-H'}
Let $u,v \in A^*$. Then $\dot{u} \mathrel{S_{H'}} \B_1.\beta.v$
for some $\beta \in A'^+$ iff
$u \mathrel{[ \bar{h}.\subword_A.(S_G.\bar{h})^* ]} v$.
\end{lemma}

We are nearly done. There only remains to compose $H'$ with a LTr
that checks for the presence of $\B_1.\beta$ (and then erases it).
For this last step, we shall use further dotted copies $\ddot{\Sigma}$,
$\dddot{\Sigma}$, \ldots, of the previously used symbols.

Formally, we define two new ALTr's $T_1$ and $T_2$: see
\ifthenelse{\boolean{short_version}}{
full version.
}{
App.~\ref{app-t1-t2}. 
}
The rules of $T_1$ ensure that it satisfies
\begin{gather}
\tag{$T_1$-spec}
\label{eq-T1}
u \mathrel{S_{T_1}} v \text{ iff } u = \B_1.\alpha.b_1.u'
\text{ and }
\ddot{u} \mathrel{I_{T_1^\INS}} v.
\end{gather}
Regarding $T_2$,
let  $u \in (\ddot{A} \cup \ddot{A'} \cup \{\ddot{b_1},\ddot{b_1'}\})^*$ and
$v \in \dddot{A}^*$. If	 $\ddot{u'}$ is the largest subword of $u$ such
that $u' \in A^*$, then
\begin{gather}
\tag{$T_2$-spec}
\label{eq-T2}
u \mathrel{S_{T_2}} v \text{ iff } \dddot{u'} \subword_{\dddot{A}}v.
\end{gather}
Combining \eqref{eq-T1} and \eqref{eq-T2} we obtain
\begin{gather*}
u \mathrel{S_{T_1}.S_{T_2}} v \text{ iff } u = \B_1.\alpha.b_1.u' \text{ and } \dddot{u'}
\subword_{\dddot{A}} v.
\end{gather*}
Composing these LTr's as $H'.T_1.T_2$ yields a resulting $G^{(+)}:\dot{A}\vdash
\dddot{A}$, which,  up to a bijective change of symbols, is what we need to
build to prove Theorem~\ref{theo-transitive}.


\section{Conclusion}
\label{sec-concl}

In this paper we introduce a notion of transformations computed by leftist
grammars and define constructions showing how these transformations are
effectively closed under sequential composition and transitive
closure.

These operations require that some ``typing'' assumptions are satisfied
(e.g., we only know how to build a transitive closure on leftist
transformers that are ``\emph{anchored}'') which may be seen as a lack of
elegance and generality of the theory, but which we see as an indication that
leftist grammars are very hard to control and reason about.

Anyway, the restrictive assumptions are not a problem for our purposes: we
intend to rely on the compositional foundations for building, in a modular
way, complex leftist grammars that are able to simulate lossy channel
systems. Here the modularity is
essential not so much for \emph{building} complex grammars. Rather, it is
essential for proving their correctness by a divide-and-conquer approach,
in the way we proved the correctness of our encoding of $\3SAT$ instances
in Section~\ref{sec-np-hard}.

As another direction for future work, we would like to mention that the
proof that accessibility is decidable for LGr's (see~\cite{motwani2000})
has to be fixed and completed.

\paragraph{Acknowledgements.}
Sylvain Schmitz helped tremendously with his numerous remarks and suggestions.



\ifthenelse{\boolean{short_version}}{
\bibliographystyle{plain}
}{
\bibliographystyle{alpha}
}

\bibliography{lcs}

\ifthenelse{\boolean{short_version}}{
}{

\newpage

\appendix

\pagestyle{myheadings}
\pagenumbering{roman}
\markboth{{\bf Technical appendix, not for the proceedings version.}}{{\bf Technical appendix, not for the proceedings version.}}



\section{Greedy derivations are sufficient}
\label{app-prop-greedy}

The lexicographic ordering between derivations is denoted $\leqlex$.

\begin{lemma}
\label{lem-has-equiv-mini}
Every derivation has a $\mu$-minimal equivalent.
\end{lemma}
\begin{proof}
Direct from observing that $\leq_\mu$ is a well-founded quasi-ordering over
derivations. Indeed, while $\leqlex$ is not well-founded over the set of
tuples of natural numbers, it is well-founded over the set
$\bigcup_{n\in\Nat}\{n\}\times \Nat^n$ to which measures of derivations
belong.
\qed
\end{proof}

The proof of Proposition~\ref{prop-greedy} is concluded with the following:
\begin{lemma}
\label{lem-mini-is-greedy}
A $\mu$-minimal derivation is greedy.
\end{lemma}
\begin{proof}
By combining
lemmas~\ref{lem-mini-is-leftmost}, \ref{lem-mini-is-eager}
and~\ref{lem-mini-is-pure} below.
\qed
\end{proof}

\begin{lemma}
\label{lemaux1}
Assume $\pi=u_0\stepgram^{r_1,p_1}u_1\stepgram^{r_2,p_2}u_2$ is a
two-step derivation. If $p_2<p_1-1$, or if $p_2=p_1-1$ and $r_1$ is an
insertion rule, then $\pi$ is not $\mu$-minimal.
\end{lemma}
\begin{proof}
The hypothesis ensures that the two steps do not interfere. Thus they can
be swapped, yielding an equivalent derivation
$\pi'=u_0\stepgram^{r_2}\stepgram^{r_1}u_2$. Clearly, $\pi'<_\mu\pi$.
\qed
\end{proof}

In the rest of this section, we consider a generic transformation $\pi$ of the
form $u_0\stepgram^{r_1,p_1} u_1 \stepgram^{r_2,p_2} u_2
\cdots\stepgram^{r_n,p_n} u_n$ in the context of some LGr
$G=(\Sigma,P,g)$.

\begin{lemma}
\label{lem-mini-is-leftmost}
A $\mu$-minimal derivation is leftmost.
\end{lemma}
\begin{proof}
Assume $\pi$ is not leftmost. Then it contains a step
$u_{i-1}=w_1 a w_2 \stepgram w'_1 a w_2=u_i$ where $w_1$ is not inert in
the rest of $\pi$. Let $j>i$ be the first step after $i$ where a letter of
$w_1$ is active: the subderivation $u_{j-2}\stepgram^{r_{j-1},p_{j-1}}
u_{j-1} \stepgram^{r_j,p_j}$ has $p_j<p_{j-1}$, and even $p_j<p_{j-1}-1$ if
$r_{j-1}$ is a deletion rule. Lemma~\ref{lemaux1} applies and entails that
this subderivation, hence also $\pi$, is not $\mu$-minimal.
\qed
\end{proof}

\begin{lemma}
\label{lem-mini-is-eager}
A $\mu$-minimal derivation is eager.
\end{lemma}
\begin{proof}
Assume $\pi$ is not eager. Let $u_{i-1}\stepgram^{r_i,p_i}u_i$ be the first
step that violates eagerness: then $u_{i-1}$ is some $w_1 b a w_2$, $w_1b$
will remain inert in the rest of $\pi$, and $b$ will eventually be deleted
at some step $j>i$, but this is not done right now even though $P$ contains
$a\del b$.

We now consider several cases for step $i$. If the active letter occurs to
the right of $a$, then one obtains a new derivation $\pi'$ by
deleting $b$ (using $a\del b$) right now, continuing like $\pi$, and
skipping step $j$ since $b$ has already been deleted. This produces an
equivalent derivation, with
$\mu(\pi')=\tuple{n,p_1,\ldots,p_{i-1},l,\ldots}$ where
$l=\size{w_1b}<p_i$. Hence $\pi$ is not $\mu$-minimal. If $a$ is the
active letter, the step must be an insertion $a\ins c$ and $p_{i+1}\geq
p_i$: we obtain, as in the previous case, an equivalent $\pi'$ with
$\mu(\pi')=\tuple{n,p_1,\ldots,p_i,p_i-1,\ldots}\ltlex\mu(\pi)$.
Finally, the active letter cannot be to the left of $a$ since $w_1b$
remains inert.
\qed
\end{proof}

\begin{lemma}
\label{lem-mini-is-pure}
A $\mu$-minimal derivation is pure.
\end{lemma}
\begin{proof}
Assume $\pi$ is not pure. Then it inserts at some step a useless letter $a$
that stays inert and is eventually deleted. By not inserting $a$ and not
deleting it later, one obtains an equivalent but shorter derivation.
\qed
\end{proof}



\section{Proof of the Closure Property (Prop.~\protect\ref{prop-closure})}
\label{app-prop-closure}

Let $G=(\Sigma,P,g)$ be some arbitrary LGr.  The following two observations
are easy.
\begin{fact}
\label{fact-cl1}
Assume $u a u'g\stepgram_G^* vg$ is a derivation where the letter $a$ is inert and
eventually erased. Then $u a^n u'g \stepgram_G vg$ for all $n\in\Nat$.
\end{fact}
\begin{fact}
\label{fact-cl2}
Assume $a$ does not occur in $u$.  Then $ug \stepgram_G^* v a v' g$ iff
 $ug \stepgram_G^* v aa v' g$.
\end{fact}

Now, in the case where $G$ is an LTr $(A,D,P,g)$, we deduce
$\invsubword_A.\approx.R_G\subseteq R_G$ from Fact~\ref{fact-cl1}, and
$R_G.\approx \: \subseteq R_G$ from Fact~\ref{fact-cl2}. This entails $R_G
\: = \: \invsubword_A.\approx.R_G.\approx$.
\qed


\section{Proof of the Composition Lemma (Lemma~\protect\ref{lem-composition})}
\label{app-lem-composition}

There only is to prove that $R_{G_1.G_2}\subseteq R_{G_1}.R_{G_2}$. For this we
consider a greedy
derivation $\pi = (ug \stepgram_{G_1.G_2}^* vg)$
with $u\in A_1^*$ and $v\in C_2^*$, and consider two cases:
\\
1.\  If $\pi$ never uses a rule from $G_2$, then no symbols from $C_2$ are
inserted and necessarily $v=\epsilon$. We obtain $u R_{G_1}.R_{G_2}v$ by
observing that $u R_{G_1} \epsilon$ (as witnessed by $\pi$) and that
$\epsilon R_{G_2}\epsilon$ (true of all leftist transformations).
\\
2.\  Otherwise, we isolate the first ${G_2}$ step
and write $\pi$ under the form
\[
\obracew{ug\stepgram_{G_1}^*wg}{\pi_1}
	\stepgram_{G_2}
\obracew{w'g\stepgram_{G_1.G_2}^*ug}{\pi_2}.
\]
Necessarily, $w\in A_1^*D_1^*$ (Fact~\ref{LTr-basic-invariant}) and, since
symbols from $A_1\cup D_1$
are inactive in $G_2$,
the first $G_2$ step is an
insertion by $g$, i.e., some $wg\stepgram weg=w'g$ with $e\in D_2=B_2\cup
C_2$. Since $\pi$ is greedy, $w$ is inert in $\pi_2$.

Now, every word along $\pi_2$ is some $xyg$ with $x$ an inert prefix of $w$
and $y\in D_2^* = (B_2\cup C_2)^*$. This claims holds at the first step
(since $e\in D_2$) and is proved by induction for the next steps. Assume
that the $k$-th step is some $xyg\stepgram zg$: since $x$ is inert, the
active letter is in $y$, hence in $D_2$ (by ind.\  hyp.) and the step is a
$G_2$ step. If the step is a deletion step, of the last letter in $x$ or of
some letter in $y$, $zg$ satisfies the claim. If the step is an insertion
step, the claim is satisfied again since $G_2$ can only insert letters from
$D_2$ and to the right of $x$.

Finally, $\pi$ must have the form
$ug\stepgram_{G_1}^*wg\stepgram_{G_2}^*vg$. Since symbols from $A_1\cup
B_1$ cannot be erased by $G_2$ rules, then necessarily $w\in C_1^*$. Hence
$u R_{G_1}w$ and $w R_{G_2} v$, proving $u R_{G_1}.R_{G_2}v$.
\qed


\section{Proof of Lemma~\protect\ref{lemma-nabla}}
\label{app-lemma-nabla}

The inclusion $\nabla\!_G.I_{G^\INS} \subseteq R_G$ is clear in view of
$\nabla\!_G\subseteq R_G$ and since $u \mathrel{I_{G^\INS}} v$ implies
$u,v\in C^*$ and $ug\stepgram_G^* vg$.

 $I_{G^\INS} = R_{G^\INS}\subseteq R_G\:\cap\: (C^*\times C^*)$
and 

For the other inclusion, $R_G \subseteq \nabla\!_G.I_{G^\INS}$, we consider a
greedy derivation 
\[
u g = w_0 g \stepgram_G w_1 g \stepgram_G \cdots \stepgram_G w_l g = v g.
\]
Every $w_i$ is some $u_iv_i$ with $u_i\in A^*$ and $v_i\in C^*$
(Fact~\ref{LTr-basic-invariant}) and, since $A$ is inert in $G$, $u_i$ is a
prefix of $u$ (so that we can write $u$ under the form $u_i.u'_i$). Let $k$
be the first index s.t.\ $u_k=\epsilon$, so that $w_i\in C^*$ for all
$i=k,\ldots,l$ and $w_k \mathrel{I_{G^\INS}} v$. There remains to show $u
\mathrel{\nabla\!_G} w_k$, i.e., $u'_k\mathrel{\nabla\!_G} v_k$.

For this we show more generally that $u'_i\mathrel{\nabla\!_G}v_i$ for all
$i=0,\ldots,k$. We proceed by induction on $i$. The base case clearly holds
since $u'_0=v_0=\epsilon$. For the induction step, we assume
$u'_i\mathrel{\nabla\!_G}v_i$ and a witness
$h_i:\{1,\ldots,\size{u'_i}\}\rightarrow\{1,\ldots,\size{v_i}\}$ for some
$i<k$. Consider the step $u_i v_i g\stepgram^r u_{i+1}v_{i+1}g$. There are
two cases.\\
(1) If $r=c\ins c'$ is an insertion rule, then the insertion must take
place in front of $v_i$ otherwise the derivation is not leftmost (the first
letter in $v_i$ cannot be inactive since it cannot be deleted and
$u_i\not=\epsilon$ must be deleted). Hence $v_{i+1}=c'v_i$,
$u'_{i+1}=u'_i$, and  a witness $h'$ for
$u'_{i+1}\mathrel{\nabla\!_G}v_{i+1}$ is obtained from $h$ with $h'(i)\egdef
h(i)+1$.
\\
(2) If $r=c\del a$ is a deletion rule, then $u_i=u_{i+1}a$, i.e.,
$u'_{i+1}=au'_i$, and $u'_{i+1}\mathrel{\nabla\!_G}v_{i+1}(=v_i)$ is
witnessed by $h'$ defined as $h'(1)\egdef 1$ and $h'(i+1)\egdef h(i)$.
\qed


\section{Proofs for Section~\protect\ref{sec-transitive}}
\label{app-transitive}

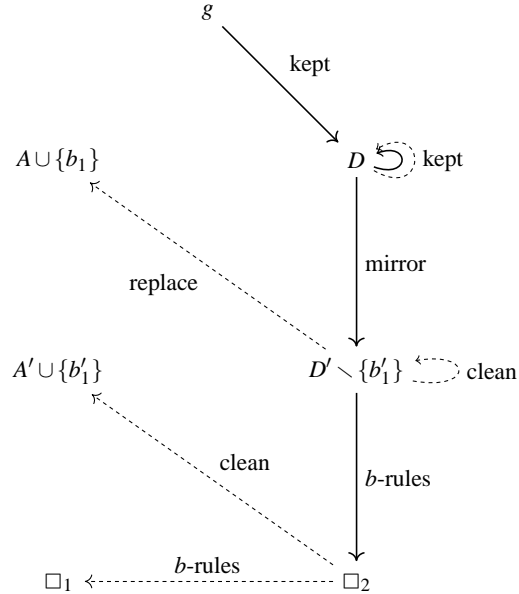
\begin{figure}[htbp]
\centering
{\setlength{\unitlength}{0.60mm}

\begin{tikzpicture}[-arcs,shorten >=1pt,auto,node distance=2.8cm,
                    semithick]

  \node (g) {$g$};
  \node (D) [below right of=g] {$D$};
  \node (D') [below of=D] {$D' \setminus \{b'_1\}$};
  \node (A) [below left of=g] {$A \cup \{b_1\}$};
  \node (A') [below of=A] {$A' \cup \{b_1'\}$};
  \node (B2) [below of=D'] {$\B_2$};
  \node (B1) [below of=A'] {$\B_1$};
  
  \path (D) edge node {mirror} (D')
        (D') edge node {$b$-rules} (B2)
        (g) edge node {kept} (D)
        (D) edge [in=20,out=340,loop] (D);

  \path[-arcs,dash pattern=on 1.5pt off 1.5pt,line width=0.4pt]
        (D') edge  node {replace} (A)
        (B2) edge node[swap] {$b$-rules} (B1)
        (B2) edge node [swap] {clean} (A')
        (D) edge [in=30,out=330,loop] node [swap] {kept} (D)
        (D') edge [in=10,out=350,loop] node [swap] {clean} (D');

\end{tikzpicture}
}
\caption{A schematic type for $F_{G,\B_1,\B_2}$.}
\label{fig-type-F}
\end{figure}

\subsection{Proof of Lemma~\protect\ref{lem-tr1}}
\label{app-lem-tr1}

\textbf{1.} We first prove, by induction on the length of a derivation
$ug\stepgram^+_G v_1 v_2 g$, with $v_1\in (A+b_1)^*$ and $v_2\in
(D+b_2)^*$, that it can be mimicked as
$\B_1.\alpha.u\stepgram^+_{F_{G}}\B_1.\alpha.v_1.\gamma.v_2.g$
for some $\gamma\in D' \setminus \{b_1'\}$. First, and since
$F_G$ contains rules ``kept'' from $G$, it can mimic any
$G$-step that does not delete a letter from $A + b_1$. For steps
$v_1.a.d.v_2 \stepgram_G v_1.d.v_2$ using $d\del a$ with $a \in A + b_1$
and $d \in D$ then, for all $\gamma \in (D' \setminus \{b_1'\})^*$, the
following derivation exists:
\begin{align*}
  & \B_1.\alpha.v_1.a.\gamma.d.v_2 
 \stepgram_{F_G}^{mirror} \B_1.\alpha.v_1.a.\gamma.d'.d.v_2 
\\
\stepgram_{F_G}^{clean^{|\gamma|}} 
& \B_1.\alpha.v_1.a.d'.d.v_2 
\stepgram_{F_G}^{replace} \B_1.\alpha.v_1.d'.d.v_2.
\end{align*}

Once $v_1=\epsilon$, there remains to show that
 $\B_1.\alpha.\gamma.v.g\stepgram^*_{F_G}\B_2.\beta.v.g$ for
 some  $\beta\in C'+b'_2$. This means erasing $\alpha$, and 
 replacing $\gamma$ (that may belong to $D'\setminus C'$), using a primed
 version of the first symbol of $v$. Formally, we use
\begin{align*}
\B_1.\alpha.\gamma.v.g 
& \stepgram_{F_G}^{mirror} \B_1.\alpha.\gamma.\beta.v.g  
  \stepgram_{F_G}^{clean} \B_1.\alpha.\beta.v.g 
\\
& \stepgram_{F_G}^{b\text{-}rules} \B_1.\alpha.\B_2.\beta.v.g
  \stepgram_{F_G}^{clean^{|\B_1.\alpha|}} \B_2.\beta.v.g.
\end{align*}

\textbf{2.}
  If $F_G$ uses a kept rule, then the same rule exists in $G$
  and is also usable. If a mirror, clean, or $b$-rule is used then $G$
  can mimic by  doing nothing. If a replace rule occurs in a step of the form
  $\B_1.\alpha.v_1.a.d'.v_2 \stepgram_G^{replace}
  \B_1.\alpha.v_1.d'.v_2, a \in A, d' \in D', v_1 \in (A + b_1)^*, v_2
  \in (D \setminus \{ b_1\})^+ $ then in a previous step there was the
  letter $d$ at the head of the $(D \setminus \{ b_1\})^+$ part of the
  word (in order to insert $d'$) and $a$ was not deleted. At that time $G$ could
  delete $a$ to finally reach $v_1.v_2$.
\qed

\subsection{Proofs for the Correctness of $H$}
\label{ssec-app-H}

\begin{figure}[htbp]
\centering
{\setlength{\unitlength}{0.60mm}

\begin{tikzpicture}[-arcs,shorten >=1pt,auto,node distance=2.8cm,
                    semithick]

  \node (g) {$g$};
  \node (D) [below right of=g] {$D$};
  \node (D') [below of=D] {$D' \setminus \{b'_1\}$};
  \node (A) [below left of=g] {$A \cup \{b_1\}$};
  \node (A') [below of=A] {$A' \cup \{b_1'\}$};
  \node (B2) [below of=D'] {$\B_2$};
  \node (B1) [below of=A'] {$\B_1$};
  
  \path (D) edge node {mirror} (D')
        (D') edge node {$b$-rules} (B2)
        (g) edge node {kept} (D)
        (D) edge [in=20,out=340,loop] (D);

  \path[-arcs,dash pattern=on 1.5pt off 1.5pt,line width=0.4pt]
        (D') edge  node {replace} (A)
        (B2) edge [bend left=10] node[swap] {$b$-rules} (B1)
        (B2) edge node [swap] {clean} (A')
        (D) edge [in=30,out=330,loop] node [swap] {kept} (D)
        (D') edge [in=10,out=350,loop] node [swap] {clean} (D');

  \path (A) edge node [swap] {mirror} (A')
        (A') edge node [swap] {$b$-rules} (B1)
        (g) edge node [swap] {kept} (A)
        (A) edge [in=173,out=187,loop] (A);

  \path[-arcs,dash pattern=on 1.5pt off 1.5pt,line width=0.4pt]
        (A') edge (D)
        (B1) edge [bend left=10] (B2)
        (B1) edge (D')
        (A) edge [in=170,out=190,loop] node {kept} (A)
        (A') edge [in=170,out=190,loop] node {clean} (A');

\end{tikzpicture}
}
\caption{Type of $H_G$.}
\label{fig-type-H}
\end{figure}
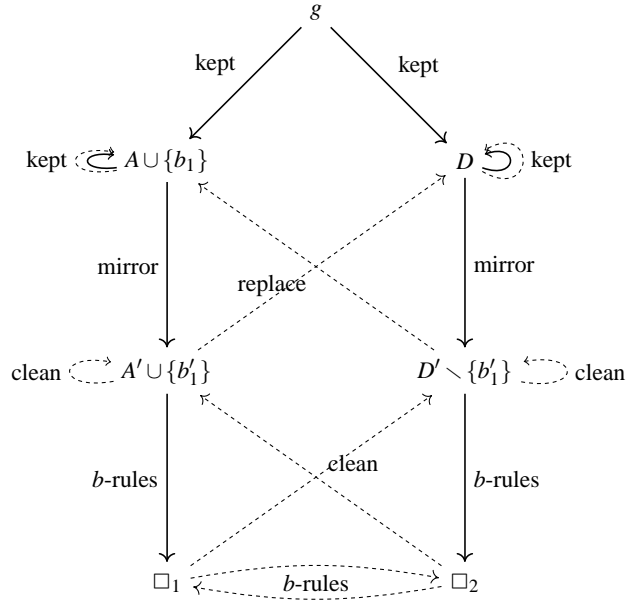

Let $\alpha, \alpha' \in (A'+b_1')^+$ and $u,v \in (A+b_1)^*$. 
\begin{lemma}
\label{nobb}
If the
  derivation $r = \B_1.\alpha.u.g \stepgram_{H}^* \B_1.\alpha'.v.g$
  is greedy, then there are no insertions of a $\B_1$ to the right of
  another $\B_1$ or of a letter  $\in(A'\cup b_1')$, and there also are no insertions of
a  $\B_2$ to the right of another $\B_2$ or a of letter  $\in(D'\setminus \{b_1\})$.
\end{lemma}
\begin{proof}[Idea]
  If there is that kind of insertion, then the $\B$ is kept until the
  end of the derivation or deleted. Since that letter can't insert or
  delete in that position, then if it is kept, there is always some letter
  on its left. There is only one $\B$ at the end of the derivation, so
  the letter is deleted. That $\B$ has no descendant, is not present
  at the end and did not delete anything, so it is useless, which
  contradicts the greediness hypothesis.
\end{proof}

\begin{lemma}
\label{invara'b}
\label{invaraa'}
The following three languages are invariants of $H$:
\begin{align}
\tag{$I_1$}
& \Sigma^*.(A + D).\Sigma_\B.\Sigma^*,
\\
\tag{$I_2$}
& \Sigma^*.(A+b_1).( \Sigma_\B + A' + b_1' ).\Sigma^*,
\\
\tag{$I_3$}
& \Sigma^*.(D\setminus b_1).( \Sigma_\B + D' \setminus
  b_1').\Sigma^*.
\end{align}
\end{lemma}
\begin{proof}
We only prove the first invariant, the other two rely on similar arguments.
  Let $w = u.l.\B.v$ with $l \in A \cup D$ and $\B \in \Sigma_\B$.  The
  invariance of $I_1$
  could be violated by inserting a symbol between $l$ and $\B$, or
  deleting $l$ or $\B$. $\B$ cannot delete $l$ or insert letters, and
  $\B$ can only be deleted by a letter from $\Sigma_\B$, a situation where the
  invariant is preserved.
\qed
\end{proof}


We say that a word $a.w$ \emph{blocks} a language $L\subseteq \Sigma^*$ if
  for all $v\in L$, for all $u\in \Sigma^*$, for all derivations
$\pi = (u.a.w.v.g \stepgram^* xg)$ where $uw$ is inert, $u.a$ is not deleted.
The definition means that no $vg$ with $v\in L$ can erase anything left of $w$.
Obviously, since we only consider derivations with $uw$ inert, any $w'$
with $a.w\subword w'$ blocks $L$ when $a.w$ does.

The arguments used to prove Lemma~\ref{invaraa'}
can be reused to show
the following:
\begin{lemma}
\label{lem-blocks}
For all $a \in (A \cup \{b_1\})$ and all $d \in (D \setminus \{b_1\})$,
\begin{align*}
&\text{$a$ blocks $(A'+b_1'+\B_1+\B_2).\Sigma^*$},
\\
&\text{$d$ blocks $(D'\setminus \{b_1'\} +\B_1+\B_2).\Sigma^*$},
\\
&\text{$a.d$ blocks $(A \cup A' \cup \{b_1,b_1',\B_1\}).\Sigma^*$},
\\
&\text{$d.a$ blocks $(D \cup D'\setminus \{ b_1, b_1'\} \cup \{\B_2\}).\Sigma^*$},
\\
&\text{$a.d'$ blocks $(A' \cup \{b_1'\}).\Sigma^*$},
\\
&\text{$d.a'$ blocks $(D' \setminus \{b_1'\}).\Sigma^*$},
\\
&\text{$a.d'.d.a$ blocks $\Sigma^*$},
\\
&\text{$d.a'.a.d$ blocks $\Sigma^*$}.
\end{align*}
\end{lemma}

We now have the necessary tools to prove
that a greedy derivation by $H$ alternates between two specific modes where
there is no confusion between $G$ steps and steps by the renamer:
We define $L_{AC},L_{CA}\subseteq \Sigma_g^*$ as
\begin{align*}
L_{AC} &= \{\B_1.(A' + b_1')^{n_1}.\B_2^{n_2}.(A+b_1)^{n_3}.(D' \setminus
b_1')^{n_4}.(D \setminus b_1)^{n_5}.g ~|~ \text{s.t.\  }\Ccal \},
\\
L_{CA} &= \{\B_2.(D'\setminus b_1')^{n_1}.\B_1^{n_2}.(D \setminus
b_1)^{n_3}.(A'+ b_1')^{n_4}.(A + b_1)^{n_5}.g~|~\text{s.t.\ }\Ccal \},
\end{align*}
where numbers $n_1,n_2,\ldots,n_5\in\Nat$ can take any values that respect
the $\Ccal$ constraint:
\begin{align}
\notag
                                         &n_2 \leq 1,
\\
\notag
\wedge \text{ if } n_1 = 0 \text{ then } &n_2 > 0,
\\
\tag{$\Ccal$}
\wedge \text{ if } n_2 > 0 \text{ then } &n_3 = 0,
\\
\notag
\wedge \text{ if } n_3 = 0 \text{ then } &n_4 > 0,
\\
\notag
\wedge \text{ if } n_4 > 0 \text{ then } &n_5 > 0.
\end{align}

Let $\alpha, \beta \in A'^+$, and $u,v \in A^*$.
\begin{lemma}
\label{lem-lac-lca}
If there is a derivation $ \B_1.\alpha.u.g \stepgram^*_{H} \B_1.\beta.v.g$,
then there is an equivalent greedy one such that each step is in $L_{AC}
\cup L_{CA}$.
\end{lemma}
\begin{proof}
Consider the greedy derivation with the fewest
  number of steps out of $L_{AC} \cup L_{CA}$. 

We assume that this number is $>0$ and obtain a contradiction, thus proving
  the Lemma. For this we consider the first step that goes out of $L_{AC}
  \cup L_{CA}$. Assume that this step is $wg\stepgram_H \ldots$ for $w.g
  \in L_{AC}$ (hence $w$ can be written under the form
  $\B_1.X.\B_2^{n_2}.Y.Z.T.g$) and proceed by a case analysis of which
  letter is active in this step:
  \begin{itemize}
  \item $\B_1$ has no rule that can be applied here.
  \item some letter $x\in A' \cup b_1'$ in $X$:
    \begin{itemize}
    \item $x \del a'$, $a' \in A' \cup b_1'$ stays in $L_{AC}$,
    \item $x \del c$, $c \in C$ is not usable,
    \item $x \ins \B_1$ forbidden by
      Lemma~\ref{nobb} (whether the insertion is at the head of $X$ or
      inside it),
    \end{itemize}
  \item $\B_2$:
    \begin{itemize}
    \item $\B_2 \del \B_1$ leads to $L_{CA}$ (since necessarily in this
      case $n_1 = 0$, $n_2 = 1$, $n_3 = 0$, $n_4 > 0$, and $n_5 > 0$),
    \item $\B_2 \del a'$, $a'\in A' \cup b_1'$ stays in $L_{AC}$,
    \end{itemize}
  \item some letter  $x\in A \cup b_1$ in $Y$:
    \begin{itemize}
    \item $x \del a$ or $x \ins a$, $a \in A \cup b_1$ stays in $L_{AC}$,
    \item $x \ins a'$,  $a' \in A' \cup b_1'$ stays in $L_{AC}$ if inserted
      at the head of $Y$, is forbidden by Lemma~\ref{nobb} if inside
      of $Y$,
    \end{itemize}
  \item some letter $x\in D' \setminus b_1'$ in $Z$:
    \begin{itemize}
    \item $x \del d'$, $d' \in D' \setminus b_1'$ stays in $L_{AC}$,
    \item $x \del a$, $a \in A \cup b_1$ stays in $L_{AC}$,
    \item $x \ins \B_2$ if $n_3 >0$ or inserts inside $Z$, forbidden by
      Lemma~\ref{invara'b}, else stays in $L_{AC}$ if $n_2 = 0$,
      else forbidden by Lemma~\ref{nobb},
    \end{itemize}
  \item some letter $x\in D\setminus b_1$ in $T$:
    \begin{itemize}
    \item $x \del d$, or $x \ins d$, $d \in D\setminus b_1$ stays in $L_{AC}$,
    \item $x \ins d'$, $d' \in D'\setminus b_1'$ stays in $L_{AC}$ if
      inserts at head of $T$,  else forbidden by Lemma~\ref{invaraa'},
    \end{itemize}
  \item $g$:
    \begin{itemize}
    \item $g \ins d$, $d \in D\setminus b_1$ stays in $L_{AC}$,
    \item $g \ins a$, $a \in A\cup b_1$:
      \begin{itemize}
      \item if $n_5 = 0$ then $n_4 = 0$ and $n_3 > 0$ so the step  stays in $L_{AC}$,
      \item if $n_5 > 0$:
        \begin{itemize}
        \item if $n_4 = 0$ then $n_3 > 0, n_2 = 0$ and $n_1 > 0$:
          
          Then $w = \B_1.X.Y.T$ with $\size{T}>0$. We write $T=yv$ with $y
          \in D\setminus b_1$ and let $w' = w.a$. After the insertion of
          $a$, $w$ is inert (by leftmost). Since $y$ does not appear at the
          end of the derivation, it is deleted.

          If $y$ is deleted by a descendant of a letter of $A \cup b_1$, then
          it is by some $a'\in A'+b_1'$. The word at the step just after the
          deletion would start with $\B_1.\alpha.u.a'$  which is
          forbidden by Lemma~\ref{invaraa'}.
          
          The only other letters able to delete $y$ are letters from
          $D\setminus b_1$. When $g$ inserts the first $d \in D\setminus b_1$ after $a$, the
          letter at its left is a $a$. So there is $d.a$ as subword of
          the inert letters. Since $d.a$ blocks $(D\setminus b_1).\Sigma^*$, the $d$
          won't be deleted.

        \item if $n_4 > 0$
          \begin{enumerate}
          \item if $n_3 > 0$ then $n_2=0$ and  $w$ is some
$\B_1.X.Y.Z.T$
with non-empty $X$, $Y$ and $Z$. We write $Z=yZ'$ with $y\in D'\setminus
b'_1$ and $w'=wa$.

$\alpha.u.y'.v'.v, \alpha
            \in (A'+b_1')^+, u \in (A+b_1)^+, y' \in D'\setminus b_1',
            v' \in (D'\setminus b_1')^*, v \in (D\setminus b_1)^*$ and
            $w' = w.a$. 

After the insertion of $a$, $w$ is inert (by leftmost). Since $y$ does not
            appear at the end of the derivation, it is deleted.

            It cannot be deleted by descendants of $A \cup b_1$, which
            could be $\B_1$ because $(A+b_1).\B_1$ is a forbidden invariant
            (Lemma~\ref{invaraa'}).

            $\B_1.\alpha.u$ cannot be deleted since $u.y'.v.a$ blocks
            $\Sigma^*$ (Lemma~\ref{lem-blocks}). Hence letters from
            $D' \setminus b_1'$ are eventually deleted by some $\B_1$
            which is not the one on the left and all the accessible
            words are in $\B_1.\alpha.u.\Sigma^*$\-$.(\B_1 +
            \B_2).\Sigma^*.g$, to which $\B_1.\beta.v.g$ does not belong.

          \item if $n_3 = 0$
            \begin{enumerate}
            \item if $n_2 = 0$ then $n_1 > 0$ and $w$ is some
$\B_1.X.Z.T$ with non-empty $X$ and $Z$. We write $Z=yZ'$ with $y\in
	      D'\setminus b'_1$ and $w'=wa$.
              
              The $\B_1.\alpha$ prefix is eventually deleted: the
              letters from $D'\setminus b_1'$ cannot be deleted without introducing a
              $\B_1$ to the right of the first $\B_1$. Since in the
              last word, there is only one $\B_1$, at least one of the
              two is deleted. Since $\Sigma^*.(\B_1 + \B_2).\Sigma^*$
              is an invariant, if the leftmost $\B_1$ is not deleted,
              the word will stay in $\B_1.\Sigma^*.(\B_1 +
              \B_2).\Sigma^*$.
              
              So there is a subderivation of the form
\begin{align*}
\B_1.X.y.Z''.g \stepgram^{r_0}
&              \B_1.X.y.Z'.T.a.g \stepgram^{r_1} \B_2.x.g
\\
              \stepgram^{r_2} &\B_1.x'.g \stepgram^{r_3} \B_1.\beta.T.g
\end{align*}
              such that $\B_1.x'.g$ is the step where the $\B_2$ from
              $\B_2.x.g$ is deleted and $\B_1.X.y.Z''.g$ is the
              last step where $y$ is active or when it was inserted.

              We can transform the sub-derivation such that
              $\B_1.X.y.Z''.g \stepgram \B_1.X.\B_2.y.Z''.g
              \stepgram^{|X| + 1} \B_2.y.Z''.g \stepgram^{r_0}
              \B_2.y.Z.T.a.g \stepgram^{r_1'} \B_2.\B_2.x.g
              \stepgram^{r_2} \B_2.\B_1.x'.g \stepgram \B_1.x'.g
              \stepgram^{r_3} \B_1.\beta.T.g$,  where $r_1'$ is $r_1$
              without the steps deleting letters from $\B_1.X$.

              That derivation is still greedy, and has strictly fewer
              steps out of $L_{AC} \cup L_{CA}$ than the original.

            \item if $n_2 > 0$ and $n_1 > 0$ then the rule $\B_2 \del
              a, a'\in A' \cup b_1'$ has priority over the insertion (by eagerness),

            \item if $n_2 > 0$ and $n_1 = 0$ then the rule $\B_2 \del
              \B_1$ has priority over the insertion (by eagerness).

            \end{enumerate}
          \end{enumerate}
        \end{itemize}
      \end{itemize}
    \end{itemize}
  \end{itemize}

  The case where $w.g \in L_{CA}$ is symmetrical. The only differences
  are the reasons why some letters must be deleted. In the 
  $L_{AC}$ case, when a letter $a \in A \cup A' \cup \{ b_1,b_1'\}$ is inert
  and has some letter from $D' \cup D$ at its left, it must be deleted
  to permit the deletion of the letters from $D \cup D'\setminus
  \{b_1',b_1\}$. In the $L_{CA}$  case, the equivalent letter is in $D
  \cup D'\setminus \{b_1',b_1\}$ so is not present at the end.
\qed
\end{proof}

\subsection{Proof of Lemma~\protect\ref{lem-H}}
\label{app-lem-H}

Let $\alpha, \beta \in (A' + b_1')^+$ and $u,v \in (A+b_1)^*$.
\begin{lemma}
\label{lem-exist-ui}
If there is a derivation $
  \B_1.\alpha.u.g \stepgram_{H}^* \B_1.\beta.v.g$, then there exists
some $n$ and some words  $u_1,v_1,u_2,\ldots,u_n,v_v,u_{n+1}$ such that,
for all $i\leq n$,  $u_i \in
\B_1.(A'+b_1')^+.(A+b_1)^+$,
 $v_i \in \B_2.(D'\setminus
\{b_1\})^+.(C+b_2)^+$ and 
\begin{align*}
\left\{
\begin{array}{l}
 u_i.g \stepgram_{F_G}^+ v_i.g
  \stepgram_{F_{\renaming}}^+ u_{i+1}.g,
\\
\B_1.\alpha.u.g
  \stepgram_{F_{\renaming}}^* u_1.g,
  u_{n+1}.g \stepgram_{F_{\renaming}}^*
  \B_1.\beta.v.g.
\end{array}\right. 
\end{align*}
\end{lemma}
\begin{proof}
  Let $\alpha, \beta \in (A'+b_1')^+, u,v \in (A+b_1)^*$ and a derivation $
  \B_1.\alpha.u.g \stepgram^*_{H} \B_1.\beta.v.g$. By
  Lemma~\ref{lem-lac-lca}, we know that there is a greedy derivation, say
  $w_0\stepgram w_1 \ldots w_n$,
  such that every $w_i$ is in $L_{AC} \cup L_{CA}$. 
  We first note that if $w_i.g \in L_{AC}$
  and $w_{i+1}.g \in L_{CA}$, then $w_{i+1}.g \in \B_2.(D'\setminus b_1')^+.(D\setminus b_1)^+.g$.

  Note also that if $w_i.g \in \B_2.(D' \setminus b_1')^+.(D\setminus
  b_1)^+.g$ and $w_{i+1}.g \not \in \B_2.(D' \setminus
  b_1')^+.(D\setminus b_1)^+.g$ then $w_{i+1}.g \in \B_2.(D' \setminus
  b_1')^+.(D\setminus b_1)^+.(A+b_1).g$.

  We will choose as $v_i$ all such $w_{j_i} \in \B_2.(D' \setminus
  b_1')^+.(D\setminus b_1)^+$ and $w_{j_i + 1} \in \B_2.(D' \setminus
  b_1')^+.(D\setminus b_1)^+.A$. Similarly for $u_i$, $w_{k_i} \in
  \B_1.(A'+b_1')^+.(A+b_1)^+$ and $w_{k_i + 1} \in \B_2.(D' \setminus
  b_1')^+.(D\setminus b_1)^+.(A+b_1)$.
  
  We  now see that $\forall i, j_i \leq k_i$ and $k_i \leq
  j_{i+1}$. This is directly implied by the fact that there is no way
  to go from $\B_2.(D' \setminus b_1')^+.(D\setminus b_1)^+.(A+b_1)$
  to $\B_2.(D' \setminus b_1')^+.(D\setminus b_1)^+$ in $L_{CA}$. So
  there is a step in $L_{AC}$ and the first step is in
  $\B_1.(A'+b_1')^+.(A+b_1)^+$.

  Next, we notice that since from a step $\B_2.(D' \setminus
  b_1')^+.(D\setminus b_1)^+ \stepgram \B_2.(D' \setminus
  b_1')^+.(D\setminus b_1)^+.A$ (included) to a step
  $\B_1.(A'+b_1')^+.(A+b_1)^+ \stepgram
  \B_1.(A'+b_1')^+.(A+b_1)^+.(D\setminus b_1)$ (excluded) there are
  only letters from $\B_1 + A' + A + b_1 + b_1' + g$ active and $g$
  only inserts letters from $A + b_1$. Those are rules from
  $F_{\renaming}$, i.e., $v_i.g
  \stepgram_{F_{\renaming}}^+
  u_{i+1}.g$. Conversely $u_i.g \stepgram_{F_{G}}^+ v_i.g$ by
  the same arguments.

  To conclude, we need to show that in fact $v_i \in \B_2.(C' +
  b_2')^+.(C+b_2)^+$. There are no rule in
  $F_{\renaming}$ deleting letters from $B \cup B'
  \setminus \{b_1,b_1'\}$ so if $v_i.g
  \stepgram_{F_{\renaming}}^+ u_{i+1}.g$, since
  there are no letter from $B \cup B' \setminus \{b_1,b_1'\}$ in
  $u_{i+1}.g$, then there were none in $v_{i}.g$.
\qed
\end{proof}

We can now prove Lemma~\ref{lem-H}:
consider a derivation $\B_1.\alpha.u.g
\stepgram^*_{H} \B_1.\beta.v.g$.
With Lemma~\ref{lem-exist-ui}, we can find $u_i,
  v_i$ such that $u_i \in \B_1.A'^+.A^+$, $v_i \in \B_2.D'^+.C^+$. Then
for all 
  $i \leq n$, 
\begin{align*}
& u_i.g \stepgram^+_{F_{G}} v_i.g
  \stepgram^+_{F_{\renaming}} u_{i+1}.g,
\\
\text{and }&
\B_1.\alpha.u.g
  \stepgram^*_{F_{\renaming}} u_1.g,
  u_{n+1}.g \stepgram^*_{F_{\renaming}}
  \B_1.\beta.v.g.
\end{align*}
  Let us write $u_i, v_i$ as $u_i = \B_1.\alpha.u'_i, \alpha \in A'^+,
  u'_i \in A^+$ and $v_i = \B_2.\beta.v'_i, \beta \in D'^+, v'_i \in
  C^+$.

  This means that $u_i S_{F_G} v_i$ and $v_i
  S_{F_{\renaming}} u_{i+1}$. Hence $u'_i S_{G} v'_i$ and $v'_i
  S_{\renaming} u'_{i+1}$ using
  Lemma~\ref{lem-fg=g}. With
  $S_\renaming = \: \approx.\subword.\bar{h}$ and
  $\subword_A.S_G.\subword_C \:= S_G$ we have $u'_i (S_{G}.\bar{h})
  u'_{i+1}$. So $u'_1 (S_{G}.\bar{h})^n u'_{n+1}$.

  Thus $u_1$ and $u_{n+1}$ are in $\B_1.A'^+.A^+$ and 
\[
\B_1.\alpha.u
  \stepgram^*_{F_{\renaming}}
  \B_1.\alpha'.u'_1
\;\;\;\text{and}\;\;\;
\B_1.\alpha'.u'_{n+1}
  \stepgram^*_{F_{\renaming}} v.
\]
 Since
  $\renaming$ is a simple transformer, it has no
  rules $a \del b$ for $a,b \in A$ so there are only insertion on $A$
  between $u$ and $u'_1$ and between $u'_{n+1}$ and $v$. Thus $u
  \subword u'_1$ and $u'_{n+1} \subword v$.

  This concludes the proof that $u \subword u'_1 (S_{G}.\bar{h})^n u'_{n+1} \subword
  v$.
\qed

\subsection{Proof of Lemma~\protect\ref{lem-H'}}
\label{app-lem-H'}

  Let be a greedy derivation $\dot{\B_1}.\dot{u}.g \stepgram^*_{H'}
  \dot{\B_2}.\B_1.\beta.v.g$.

  First note that $\dot{u}$ is deleted before the insertion of a
  letter from $D$ since every letter from $\dot{A}$ blocks $D$.
  
  There could be a derivation where a word from
  $\dot{\B_1}.A'^*.A^+.D.g$ is reached. But since $\dot{\B_1}$ will be
  eventually deleted by $\dot{B_2}$, it is possible to find another
  greedy derivation where a $\B_1$ is inserted which inserts
  $\dot{B_2}$ and delete $\dot{\B_1}$ before the insertion of $d$.

  The rules used before the first insertion of a letter of $D$ define
  a renaming from $\dot{A}$ to $A$. So there is $u'$ such that $u
  \subword u' (S_G.\bar{h})^* v$.
\qed

\subsection{Finals steps of the construction: $T_1$ and $T_2$}
\label{app-t1-t2}

$T_1$ is defined as
$(A \cup A', \{\B_1,b_1,b_1'\}, \{\dot{\B_2},
o_1\}, \ddot{A} \cup \ddot{A'} \cup \{\ddot{b_1},\ddot{b_1'},\ddot{\B_1}\},
\dot{\B_2}, o_1, P_1,g)$
with the following set of rules:
\begin{description}
\item all $l \del \ddot{l}$, for $l \in \{\B_1,b_1,b_1'\} \cup A \cup A'$,
\item all $g \ins \ddot{a_i}$, for $\ddot{a_i} \in \ddot{A}$,
\item all $\ddot{a_i} \ins \ddot{a_j}$, for $\ddot{a_i},\ddot{a_j} \in \ddot{A}$,
\item all $\ddot{a_i} \ins \ddot{b_1}$, for $\ddot{a_i} \in \ddot{A}$,
\item all $\ddot{b_1} \ins \ddot{a'}$, for $\ddot{a'} \in \ddot{A'} \cup \ddot{b_1'}$,
\item all $\ddot{a'} \ins \ddot{b'}$, for $\ddot{a'},\ddot{b'} \in \ddot{A'} \cup \ddot{b_1'}$,
\item all $\ddot{a'} \ins \ddot{\B_1}$, for $\ddot{a'} \in \ddot{A'} \cup \ddot{b_1'}$,
\item $\ddot{\B_1} \ins o_1$,
\item $o_1 \del \dot{\B_2}$.
\end{description}
We let the reader check that these rules ensure the satisfaction of
\eqref{eq-T1}.
\\

We further define
$T_2 \egdef (\ddot{A} \cup \ddot{A'} \cup
\{\ddot{b_1},\ddot{b_1'},\ddot{\B_1}\}, \{o_1, o_2\} , \dddot{A} ,
o_1, o_2, P_2,g)$
with the following set of rules:
\begin{description}
\item all $g \ins \dddot{a_i}$, for $\dddot{a_i} \in \dddot{A}$,
\item all $\dddot{a_i} \ins \dddot{a_j}$, for $\dddot{a_i},\dddot{a_j} \in \dddot{A}$,
\item all $\dddot{a_i} \ins o_2$, for $\dddot{a_i} \in \dddot{A}$,
\item all $\dddot{a_i} \del \ddot{a_i}$, for $\dddot{a_i} \in \dddot{A}$,
\item all $\dddot{a_i} \del l$, for $\dddot{a_i} \in \dddot{A}$ and $l \in \ddot{A'} \cup \{\ddot{b_1},\ddot{b_1'},\ddot{\B_1}\}$,
\item $o_2 \del o_1$.
\end{description}
We let the reader check that these rules ensure the satisfaction of
\eqref{eq-T2}.





}

\end{document}